\documentclass[a4paper,twocolumn,10pt, accepted=2026-02-23]{quantumarticle}
\pdfoutput=1

\usepackage[numbers,sort&compress]{natbib}

\usepackage{graphicx}
\usepackage{bm}
\usepackage{xcolor}
\usepackage{physics}
\usepackage{verbatim}
\usepackage{blkarray}
\usepackage{amsmath}
\usepackage{amsthm}
\usepackage{amssymb}
\usepackage{makecell}

\usepackage{hyperref}

\usepackage{mathtools}
\usepackage{dsfont}

\newcommand\edit[1]{#1}

\newcommand{\Mat}{\operatorname{Mat}}

\newcommand{\Var}{\operatorname{Var}}
\newcommand{\opket}[1]{\vert #1 \rangle\!\rangle}
\newcommand{\opbra}[1]{\langle\!\langle  #1 \vert}
\newcommand{\opbraket}[2]{\langle\!\langle  #1 \vert #2 \rangle\!\rangle}
\newcommand{\opketbra}[2]{\vert  #1 \rangle\!\rangle \langle\!\langle#2 \vert}
\newcommand{\e}{\mathrm e}
\newcommand{\ii}{\mathrm i}
\newcommand{\SU}{\operatorname{SU}}
\newcommand{\SO}{\operatorname{SO}}

\newcommand{\col}{\operatorname{col}}
\newcommand{\wt}{\operatorname{wt}}
\newtheorem{prop}{Prop}
\newtheorem{lem}{Lemma}
\newtheorem{conj}{Conjecture}

\begin{document}

\title{Bridging Classical and Quantum Information Scrambling with the Operator Entanglement Spectrum}

\author{Ben T. McDonough}
\email{benjamin.mcdonough@colorado.edu}
\affiliation{Department of Physics and Center for Theory of Quantum Matter, University of Colorado, Boulder, CO 80309, USA}
\affiliation{Ames National Laboratory, Ames, IA 50011, USA}

\author{Claudio Chamon}
\email{chamon@purdue.edu}
\affiliation{Department of Physics and Astronomy, Purdue University, West Lafayette, IN 47907, USA}

\author{Justin H. Wilson}
\affiliation{Department of Physics and Astronomy, Louisiana State University, Baton Rouge, LA 70803, USA}
\affiliation{Center for Computation and Technology, Louisiana State University, Baton Rouge, LA 70803, USA}

\author{Thomas Iadecola}
\email{iadecola@psu.edu}
\affiliation{Department of Physics and Astronomy, Iowa State University, Ames, IA 50011, USA}
\affiliation{Ames National Laboratory, Ames, IA 50011, USA}
\affiliation{Department of Physics, The Pennsylvania State University, University Park, PA 16802, USA}
\affiliation{Institute for Computational and Data Sciences, The Pennsylvania State University, University Park, PA 16802, USA}
\affiliation{Materials Research Institute, The Pennsylvania State University, University Park, PA 16802, USA}

\date{July 18, 2025}

\begin{abstract}
Universal features of chaotic quantum dynamics underlie our understanding of thermalization in closed quantum systems and the complexity of quantum computations. 
Reversible automaton circuits, comprised of classical logic gates, have emerged as a tractable means to study such dynamics.
Despite generating no entanglement in the computational basis, these circuits nevertheless capture many features expected from fully quantum evolutions.
In this work, we demonstrate that the differences between automaton dynamics and fully quantum dynamics are revealed by the operator entanglement spectrum, much like the entanglement spectrum of a quantum state distinguishes between the dynamics of states under Clifford and Haar random circuits.
While the operator entanglement spectrum under random unitary dynamics is governed by the eigenvalue statistics of random Gaussian matrices, we show evidence that under random automaton dynamics it is described by the statistics of Bernoulli random matrices, whose entries are random variables taking values $0$ or $1$.
We study the crossover between automaton and generic unitary operator dynamics as the automaton circuit is doped with gates that introduce superpositions, namely Hadamard or $R_x = \e^{-\ii\frac{\pi}{4}X}$ gates.
We find that a constant number of superposition-generating gates is sufficient to drive the operator dynamics to the random-circuit universality class, similar to earlier results on Clifford circuits doped with $T$ gates. This establishes the operator entanglement spectrum as a useful tool for probing the chaoticity and universality class of quantum dynamics.
\end{abstract}

\maketitle

\section{\label{sec:sec1} Introduction}

The last decade has seen an intense effort to identify universal features of chaotic quantum dynamics. 
These unifying features describe commonalities in a wide range of physical systems, with applications to the quantum properties of black holes, the nature of thermalization in far-from-equilibrium systems, and the complexity of quantum computations~\cite{mi2021information, fisher2023random, cotler2017chaos}. 
The entanglement spectrum of a quantum state is an established tool for diagnosing complexity, thermalization, and universality in quantum dynamics \cite{geraedts2016many, yang2017entanglement, chamon2014emergent, chang2019evolution, shaffer2014irreversibility, rakovszky2019signatures, chen2018universal, geraedts2017characterizing}. It is known that the full state entanglement spectrum is a more fine-grained probe of the dynamics than the entanglement entropy alone \cite{zhou2020single}. One central result of this paper is to establish the entanglement spectrum of operators evolving in the Heisenberg picture as a fine-grained tool to probe quantum dynamics.

At the level of Heisenberg-picture dynamics, quantum chaos is often characterized by two features: operator spreading and operator scrambling \cite{mi2021information, nahum2018operator, hosur2016chaos, nahum2017quantum}. 
These features describe a process by which a Heisenberg-picture operator evolves to fill the lightcone and appear suitably random regardless of the initial condition.
One way to quantify operator scrambling is to study the growth of the operator entanglement entropy (OEE) under the Heisenberg dynamics \cite{wang2002quantum, zanardi2001entanglement,
Xu2019,styliaris2021information, xu2024scrambling, pivzorn2009operator, prosen2007chaos, prosen2007operator, dowling2023scrambling}.
The OEE is defined as the information entropy of a bipartitioning of the operator between subsystems $A$ and $B$, specifically the information shared between inputs-and-outputs of A and inputs-and-outputs of B (see Fig.~\ref{fig:OES_cartoon}).
The OEE saturates to a value extensive in system size, much like the entanglement entropy of a quantum state evolved under chaotic dynamics \cite{page1993average, sanchez1995simple}.
Moreover, the OEE can be used to characterize the complexity of a quantum evolution itself, without reference to a particular state or operator evolving under that dynamics.
In fact, the OEE for a given evolution up to time $t$ for a fixed bipartition is directly related to the value at time $t$ of the OTOC between two typical operators supported within $A$ and $B$, respectively~\cite{styliaris2021information}.

The OEE by itself is not fine-grained enough to characterize the dynamics of a generic operator as fully chaotic.
One example of this is dynamics under (reversible) automaton circuits, which are composed of classical logic gates.
These circuits act as permutation operators on the computational basis, and therefore generate no entanglement when acting on states in that basis.
However, they can generate maximal entanglement when evolving generic product states~\cite{iaconis2021quantum,pizzi2024bipartite}, and also produce signatures of chaotic operator dynamics.
For example, they exhibit ballistic operator spreading and power-law broadening of operator fronts, volume-law operator entanglement \cite{iaconis2019anomalous, pizzi2024bipartite}, and chaos-like behavior of certain OTOCs \cite{chamon2022quantum, bertini2024quantum}. 
These resemblances have established quantum automaton circuits as a computationally tractable model to study chaotic dynamics \cite{iaconis2019anomalous, iaconis2021quantum}, and also to establish parallels between classical and quantum information chaos~\cite{chamon2022quantum, pizzi2022bridging}.
At the same time, automaton circuits are clearly not equivalent to generic quantum circuits, given their non-universality.
How does this difference manifest itself in the dynamics of generic operators?

\begin{table}[t!]
    \centering
    \begin{tabular}{lll}
    \hline \hline
         Evolution & Clifford & Automata \\
         \hline
         Code states & Entangling & Non-entangling \\
         Pauli-strings & Non-entangling & Entangling\\
        \hline\hline 
    \end{tabular}
    \caption{A comparison of the (Schr\"odinger) evolution of code states vs.\ (Heisenberg) evolution of generic Pauli string operators with both Clifford dynamics and Automata dynamics.}
    \label{tab:Clifford_Automata}
\end{table}
To answer this question, it is useful to consider an analogy between automaton circuits and Clifford circuits (partially summarized in Table~\ref{tab:Clifford_Automata}). 
Where automaton circuits generate no entanglement when acting on states in the computational basis, Clifford circuits generate no operator entanglement when acting on operators in the Pauli basis.
Like automaton circuits can generate extensive operator entanglement when acting on local operators, Clifford circuits can generate extensive entanglement when acting on product states.
A useful way to detect the nonuniversality of Clifford circuits at the level of quantum states is to consider the full entanglement spectrum (ES), rather than just the entanglement entropy.
For instance, the ES of an initial random product state evolved under a random Clifford circuit exhibits Poisson statistics, while a state evolved under a Haar random unitary circuit exhibits Wigner-Dyson ES statistics~\cite{chamon2014emergent}.
This property is fine-tuned: inserting a constant number of $T$ gates is sufficient to drive the ES statistics from Poisson to Wigner-Dyson~\cite{zhou2020single,haferkamp2020quantum}, even though the operator dynamics in the Pauli basis remains relatively tractable.

In this work, we complete this analogy by studying the properties of the operator entanglement spectrum (OES), i.e.~the full distribution of the operator Schmidt coefficients, rather than just the OEE. 
While the fine-grained features of the state entanglement spectrum are well-studied, characterizing dynamics using the OES of an initially local evolved operator is less explored. 
Previous work~\cite{chen2018operator} observed that features of the OES, such as the spectral form factor, qualitatively agree with the predictions of random matrix theory (RMT).
We further elucidate this connection. 
We argue that under random unitary evolution, the operator entanglement spectrum (OES) is a universal measure of the fluctuation of operator amplitudes in Pauli-string space which is insensitive to the details of the dynamics and the initial conditions. 
In the same spirit as the famous Wigner surmise \cite{wigner1951statistical}, we use this observation to derive approximate expressions for the entanglement spectra and spacing ratio distributions of local, Hermitian observables evolved under random unitary and time-reversal (TR) invariant dynamics.
This establishes the universal OES as a direct measure of ergodicity in operator space.

We then turn to consider the entanglement spectrum of automaton circuits and of local operators evolved under random automaton dynamics.
We relate the entanglement spectrum of such an operator to the singular value distribution of a Bernoulli matrix with entries sampled from $\{0,1\}$ with $p(1) = \frac{1}{d}$, where $d$ is the Hilbert space dimension. 
We argue that just as in the unitary case, the OES behaves as though the operator amplitudes are fluctuating randomly with respect to this measure. 
Thus, we show that the OES effectively discriminates between automaton dynamics, a model of classical information chaos, and random unitary dynamics, a model of quantum information chaos.

Like Clifford circuits become computationally universal with the addition of $T$ gates, automaton circuits become universal with the addition of Hadamard gates~\cite{aharonov2003simple}. 
This provides an appealing picture of universal quantum logic as universal classical logic augmented with superposition. 
Motivated by the analogy between state dynamics in Clifford circuits and operator dynamics in automata, we study the effect of adding superposition-generating (SG) gates on the OES.
We show numerically that the universal OES under unitary dynamics can be approximated exponentially closely with a finite number of SG gates, independent of system size. 
The transition we observe the in level spacing ratios of the OES when SG gates are added is even more dramatic, and matches the behavior of the state ES level spacings under Clifford circuits: any finite number of SG gates causes the OES level spacings to converge exponentially quickly in systems size to the distribution given by Haar-random unitary operators.
This creates a bridge between classical and quantum information chaos.
\par The remainder of the paper is organized as follows. In Sec.~\ref{sec:sec2}, we show that the Marchenko-Pastur distribution is the universal OES associated to local operators evolving under random unitary and orthogonal dynamics and explore the conditions for its emergence. We derive expressions for the level spacing ratios of these spectra \cite{atas2013distribution}. We analyze the evolution of the OES under random unitary circuits, which is shown to exhibit scale-invariant behavior during intermediate times and equilibrate to the universal distribution at late times. In Sec.~\ref{sec:sec3}, we carry out a similar study of the entanglement spectrum of off-diagonal operators evolving under random automaton dynamics. We relate this spectrum to the singular value distribution of Bernoulli matrices, which is quite distinct from that of random unitary matrices. We study the properties of this distribution and derive a method for computing its moments. We discover that the convergence properties of the entanglement spectrum dynamics under random local automata are qualitatively similar to the unitary case with respect to this alternative spectral distribution. Lastly, in Sec.~\ref{sec:sec4}, we study the crossover between random unitary and random automaton dynamics as automaton circuits are doped with SG gates. We find that the OES and level spacings converge with increasing density of SG gates and system size, completing the analogy to Clifford circuits and the bridge between classical and quantum information chaos.

\subsection{Mathematical Preliminaries}
\label{sec:Mathematical Preliminaries}

Let $\mathcal H_d = \mathcal H_A \otimes \mathcal H_B$ be a Hilbert space of dimension $d$, bipartitioned into subsystems $A$ and $B$ of dimensions $d_A$ and $d_B$, respectively ($d = d_A d_B$). 
The entanglement spectrum $\{\lambda_i\}$ of an operator $X$ on $\mathcal H_d$ is defined in terms of the operator Schmidt decomposition
\begin{align}
\label{eq:schmidt}
X = \sum_i \sqrt{\lambda_i}\ \mathcal U_i \otimes \mathcal V_i,
\end{align}
where $\{\mathcal U_i\}$, $\{\mathcal V_i\}$ are two sets of trace-orthonormal operators supported on $A$ and $B$ respectively (see Fig.~\ref{fig:OES_cartoon}).
The OES of $X=\sum_{\alpha,\beta}X_{\alpha\beta}\ket{\alpha}\bra{\beta}$ is equivalent to the entanglement spectrum of the vectorized operator
\begin{align}
\label{eq:vector-X}
    \opket{X} = \sum_{\alpha,\beta} X_{\alpha\beta}\ket{\alpha}\ket{\beta}^*\equiv\sum_{\alpha,\beta} X_{\alpha\beta}\opket{\alpha\beta},
\end{align}
viewed as a state in the enlarged Hilbert space $\mathcal H_d^{\otimes 2}$.
(Here $\{\ket{\alpha}\}$ is a product-state basis for $\mathcal H_d$.) This Hilbert-space isomorphism mapping operators $X$ to kets $\opket{X}$ is isometric with the trace inner product $\opbraket{X}{X} = \Tr(X^\dagger X)$.
With respect to the bipartition $\mathcal H_d = \mathcal H_A \otimes \mathcal H_B$, Eq.~\eqref{eq:vector-X} can be rewritten as
\begin{align}
\label{eq:vector-X-bipartition}
\opket{X}=\sum_{\alpha_A,\beta_A,\alpha_B,\beta_B}X_{(\alpha_A\alpha_B)(\beta_A\beta_B)}\opket{\alpha_A\beta_A}\otimes\opket{\alpha_B\beta_B},
\end{align}
where we have defined the indices $\alpha_A,\alpha_B$ and $\beta_A,\beta_B$ via $\ket\alpha = \ket{\alpha_A}\otimes\ket{\alpha_B}$ and $\ket{\beta}=\ket{\beta_A}\otimes\ket{\beta_B}$, so that $\opket{\alpha_A\beta_A}\otimes\opket{\alpha_B\beta_B}\equiv (\ket{\alpha_A}\ket{\beta_A}^*)\otimes(\ket{\alpha_B}\ket{\beta_B}^*)$. 
The matrix $X_{(\alpha_A\alpha_B)(\beta_A\beta_B)}$ is identical to $X_{\alpha\beta}$ with $\alpha=(\alpha_A\alpha_B)$ and $\beta=(\beta_A\beta_B)$ treated as ``superindices." 
The OES is then given by the squared singular values of the matrix $X_{(\alpha_A\beta_A)(\alpha_B\beta_B)}$ obtained via the partial transpose $\alpha_B\leftrightarrow\beta_A$ of $X_{\alpha\beta}$, which we denote by $X^{t_{AB}}$ (note that $X^{t_{AB}}$ is a $d_A^2\times d_B^2$ matrix).
Equivalently, the OES are the eigenvalues of the \textit{reduced superoperator} 
\begin{align}
\label{eq:reduced-super}
    \tr_{\mathcal B(H_B)}\opketbra{X}{X} = X^{t_{AB}}(X^{t_{AB}})^\dagger,
\end{align} 
where the trace is taken over $\mathcal B(H_B)$, the Hilbert space of operators on $B$. 
\begin{figure}
    \centering
    \def\svgwidth{.85\linewidth}
    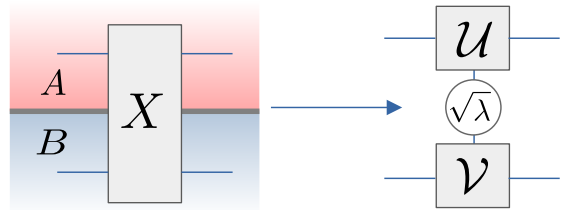
    \caption{A cartoon demonstrating the operator entanglement spectrum for an operator $X$ (rectangular box) and a bipartition into regions $A$ (red) and $B$ (blue). Lines on the left and right sides of $X$ represent registers storing the input and output states, respectively, in the Hilbert spaces $\mathcal H_A$ and $\mathcal H_B$. Performing an appropriate singular value decomposition (SVD) results in the Schmidt decomposition of $X$ as a sum over products of operators acting in subregions $A$ and $B$, Eq.~\eqref{eq:schmidt}, as shown in the diagram on the right.}
    \label{fig:OES_cartoon}
\end{figure}

Because the singular values are unique, and because the largest $k$ singular values provide the best approximation to $X$ by a sum of at most $k$ product operators on subsystems $A$ and $B$, the OES provides the most direct way to quantify the complexity with which the action of an operator on subsystem $A$ depends on subsystem $B$, or the randomness in an observable quantity when information about subsystem $B$ is lacking. 
The most commonly studied metric obtained from the OES is the operator entanglement entropy (OEE), defined as $S_{\text{op}} = -\sum_i\lambda_i\ln\lambda_i$. 
Moreover, since $\sum_i \lambda_i = \tr(X^\dagger X) = \Vert X \Vert_F^2$, after normalizing $X$, the entanglement spectrum may itself be interpreted as a probability distribution. 
The OEE is then the Shannon information entropy of this distribution. 
However, the OEE alone fails to distinguish between important classes of quantum states, such as those produced by integrable and chaotic dynamics.

The entanglement spectrum of a random \textit{state} $\ket{\psi} \in \mathcal H_d$ is known to follow the Marchenko-Pastur (MP) distribution \cite{marchenko1967distribution, collins2016random, Nadal_2011}, given by 
\begin{equation}
\text{MP}_\gamma(x) = \sqrt{(\lambda_+ - x)(x-\lambda_-)}/(2\pi x \gamma)\mathds 1_{[\lambda_-, \lambda+]} \label{eq:MP}
\end{equation}
where $\mathds 1_{[\lambda_-, \lambda_+]}$ is the indicator function of the interval $[\lambda_-, \lambda_+]$ with $\lambda_{\pm} = (1\pm \sqrt{\gamma})^2$. For a state, the parameter $\gamma$ is given by $\gamma = \frac{d_B}{d_A} < 1$. For the special case of $\gamma=1$ (i.e., $d_A = d_B$), the distribution of $\sqrt{\lambda}$ is the half-semicircle distribution $p(\sqrt{\lambda}) = \frac{1}{\pi}\sqrt{4 - \lambda}$, which we use frequently in this paper for better visualization.
The MP distribution arises as the spectral distribution or density of states (DOS) of $X X^\dagger$, called a Wishart matrix, where $X$ is a random variable taking values in $\Mat_{d_A \times d_B}(\mathbb C)$ with independent and identically distributed (i.i.d.) components. 
It is universal in the sense that this limiting distribution occurs independently of the microscopic distribution of the entries of $X$. 
Another way to understand the MP distribution is as the limiting eigenvalue distribution of a matrix of the form $\sum_{i=1}^{d_A}\ketbra{r_i}{r_i}$, where $\{\ket{r_i}\}_{i = 1}^{d_A}$ is a sample of statistically independent vectors of length $d_B$ with components distributed around zero, and $\Vert \ket{r_i}\Vert_2^2 \to \frac{d_B}{d_A} = \gamma$, where $\gamma$ remains fixed as $d \to \infty$ \cite{voiculescu1991limit, feier2012methods}.
\par In addition to the DOS, another common way to probe the spectrum is through the level spacing ratios. The level spacing ratios are defined as $r_n = \frac{\delta_{n+1}}{\delta_n}$, where $\delta_n = \sqrt{\lambda_n} - \sqrt{\lambda_{n-1}}$ and $\{\lambda_i\}$ is in decreasing order. The square roots are for easier comparison to the Wigner surmise for the classical random matrix ensembles, whose level spacing ratio distribution is given by the Wigner-Dyson (WD) formula
\begin{equation}
    p(r) = \frac{1}{Z_\beta}\frac{(r+r^2)^{\beta}}{(1+r+r^2)^{1+\frac{3}{2}\beta}}
    \label{eq:WD}
\end{equation}
where $\beta$ is known as the Dyson index and specifies the random matrix ensemble; $\beta=1$ and $\beta=2$ correspond to the Gaussian orthogonal and unitary ensembles (GOE and GUE), respectively, and $Z_\beta$ is a normalization constant. 
\edit{Unlike the MP distribution, the WD formula is not exact in the $d \to \infty$ limit. An approximate correction is found numerically in \cite{atas2013distribution}.}
The level spacing ratios depend on the joint distribution of the eigenvalues through their order statistics, and so provide more information about the spectrum than the DOS alone.

\section{\label{sec:sec2}OES in Random Unitary Dynamics}
In this section, we seek to characterize chaotic dynamics by probing the OES of local observables. First, we show that the OES can be accurately predicted by comparison with classical random matrix ensembles. In particular, after random unitary evolution, local observables have an entanglement spectrum which closely follows the spectrum of the GOE. Second, we show that these observations also hold in the more realistic setting of random unitary circuits, confirming that chaotic dynamics can be characterized using the OES of local observables. We show that the scaling behavior in a circuit of depth $l$ on $N$ qubits is an exponential function of $l/N$, approaching a scale-invariant power-law around depth $l = N/2$ and thermalizing to the MP distribution at late times.

\subsection{\label{subsec:classical_ensembles} Random evolutions and observables}
Chaotic quantum many-body evolutions can be modeled by Haar-random unitary matrices. 
To formulate an RMT description of operator dynamics under chaotic quantum evolution, we compare the OES of an initial local operator conjugated by random unitary or orthogonal matrices to classical matrix ensembles.

\subsubsection{OES of a Haar-random unitary\label{sec:Haar OES}}
The OES of a Haar-random unitary matrix $U$ [i.e., one drawn from the circular unitary ensemble (CUE)] is governed by the eigenvalues of the matrix $U^{t_{AB}}(U^{t_{AB}})^\dagger$ [see Eq.~\eqref{eq:reduced-super}].
It is known that truncations of Haar-random unitary matrices converge to Gaussian matrices in the large-dimension limit \cite{petz2004asymptotics, Zyczkowski}, making the columns of the matrix $U^{t_{AB}}$ pairwise-independent complex Gaussian variables.
This suggests that the ensemble of Wishart matrices with complex Gaussian entries is the relevant ensemble to describe the OES of $U$. 
When $d_A = d_B$, this is known as the Ginibre ensemble \cite{ginibre1965statistical}. 
If $M$ is distributed according to this ensemble, then $M M^\dagger$ is a random Hermitian matrix, whose level-spacing distribution is given by the WD formula [Eq.~\eqref{eq:WD}] with Dyson index $\beta = 2$. The OES DOS is given by the MP distribution, now with $\gamma = \frac{d_A^2}{d_B^2}$ instead of $\gamma = \frac{d_A}{d_B}$ for random states, because the space of operators has dimension $d_A^2 \times d_B^2$.

This prediction is in excellent agreement with the numerics, as shown in Fig.~\ref{fig:classical_ensembles}. 
\begin{figure}
    \centering
    \def\svgwidth{0.95\linewidth}
    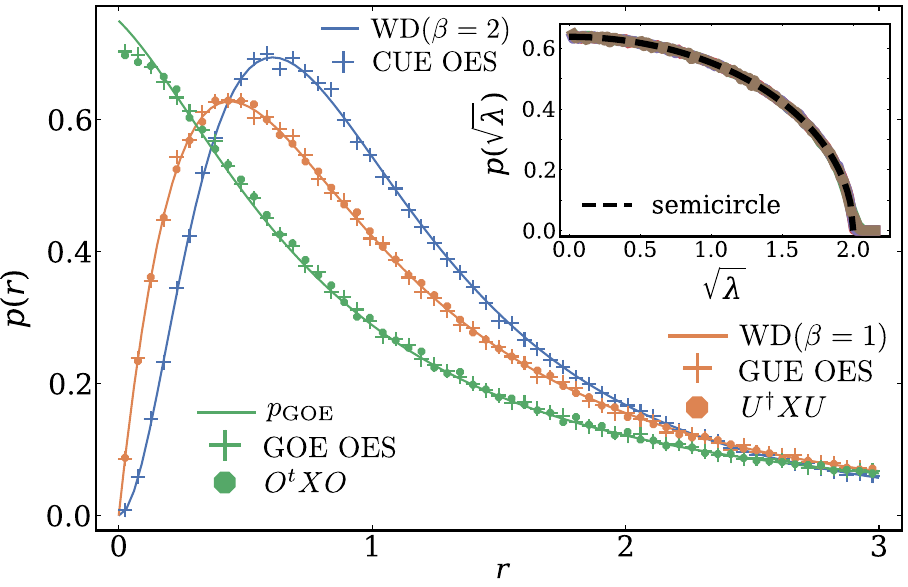
    \caption{Level-spacing ratios of ``random" observables (points) compared to those of matrices drawn from classical random matrix ensembles (crosses). The system consists of 12 sites, and the data represent 64 matrix samples. $X$ is a Pauli-$X$ operator on a single site. $U$ is drawn from the ensemble of Haar-random unitary matrices (CUE) and $O$ is drawn from Haar-random orthogonal matrices (COE). The model distribution $p_{\text{GOE}}$ for the OES spacing ratios of GOE matrices is given in Eq.~\eqref{eq:pgoe}. Inset: the entanglement spectrum DOS of each ensemble approaches the MP distribution.} 
    \label{fig:classical_ensembles}
\end{figure}
The blue line shows the Wigner-Dyson distribution with parameter $\beta = 2$ with the correction in Ref.~\cite{atas2013distribution}. 
Agreement of the DOS with the MP distribution holds for all the ensembles we will consider, and this agreement is shown in the inset, using $\{\sqrt{\lambda_i}\}$ for comparison with the more familiar semicircle.

\subsubsection{OES under chaotic operator dynamics\label{subsec:unitary-Dynamics}}
Next, we examine the entanglement spectrum of a chaotically-evolving local observable.
To model this scenario, we fix the initial operator $X$ to be the Pauli-$X$ operator supported on a single site, and consider the ensemble $\{X_U\equiv U^\dagger XU\}_U$ where $U$ is Haar-random.
The elements of this ensemble are unitary, Hermitian, traceless, and have eigenvalues $\pm 1$.
Thus it is not \textit{a priori} obvious that the OES of $X_U$ is governed by the eigenvalue distribution of Wishart matrices, whose entries are i.i.d. variables normally distributed around zero.
However, the partial transpose defined below Eq.~\eqref{eq:MP} tends to effectively destroy the correlations between the matrix elements, such that the spectrum of $X_U^{t_{AB}}(X_U^{t_{AB}})^\dagger$ reflects that of a matrix with i.i.d. entries. 
The numerical result is shown in Fig.~\ref{fig:classical_ensembles} by the orange data points. 
The orange line shows the WD level-spacing distribution Eq.~\eqref{eq:WD} with parameter $\beta = 1$, which agrees almost perfectly with the average OES of $\{X_U\}_U$, shown by the solid dots. 

We explain this agreement by comparing the OES of $X_U$ to that of an operator $M$ drawn from the GUE---\textit{i.e.}, a random observable. This is illustrated in Fig.~\ref{fig:classical_ensembles} by orange crosses.
A conceptually useful way to understand the entanglement spectrum of $M$ is as the singular values of the matrix 
\begin{align}
\label{eq:alpha}
[\alpha_{ij}] = \tr[(P_i\otimes P_j)^\dagger M] = \opbraket{P_i \otimes P_j}{M},
\end{align}
where $P_i \in \mathcal P_A, P_j \in \mathcal P_B$ are trace-normalized Pauli operators on subsystems $A$ and $B$.
As we show in Appendix~\ref{app:GUE}, the GUE is equivalently specified by sampling \textit{real} independent Gaussian variables in the Pauli basis, so $\alpha_{ij}$ are i.i.d.\ real Gaussian variables.
This implies that the OES level spacings of a GUE matrix will follow the \textit{eigenvalue} level spacings of a GOE matrix,
which is Wigner-Dyson with parameter $\beta = 1$ [Eq.~\eqref{eq:WD}]. This is a ``demotion" because the eigenvalue spacing ratio distribution of GUE matrices is WD with parameter $\beta = 2$. 
The correspondence is summarized in Table~\ref{tbl:classical_ensembles}.
\begin{table}
\centering
\begin{tabular} {c c c }
\hline\hline
Ensemble & OES DOS & OES level spacings \\
\hline
\makecell{Haar-random \\ unitary (CUE)} & MP & $\text{WD}(\beta = 2)$ \\ 
 GUE & MP & $\text{WD}(\beta = 1)$ \\  
 GOE & MP & $p_{\text{GOE}}$ \\
\hline\hline
\end{tabular}
\caption{Table showing the OES DOS and level spacing ratios of classical matrix ensembles. The level spacings of the GOE OES do not follow a Wigner-Dyson distribution, and an approximate expression for $p_{\text{GOE}}$ is given in Eq.~\eqref{eq:pgoe}.
\label{tbl:classical_ensembles}
}
\end{table}

In Fig.~\ref{fig:unitary_convergence_concentration}, we show that the OES DOS of $X_U$ agrees with the MP distribution for different bipartitions, \textit{i.e.}, for multiple values of the parameter $\gamma=\frac{d_B^2}{d_A^2}$. 
This agreement also appears to be self-averaging; as shown in the inset, for sufficiently large Hilbert space dimension $d$ the OES distribution of a single sample from $\{X_U\}_U$ closely approaches the average.
\begin{figure}
    \centering
    \def\svgwidth{0.95\linewidth}
    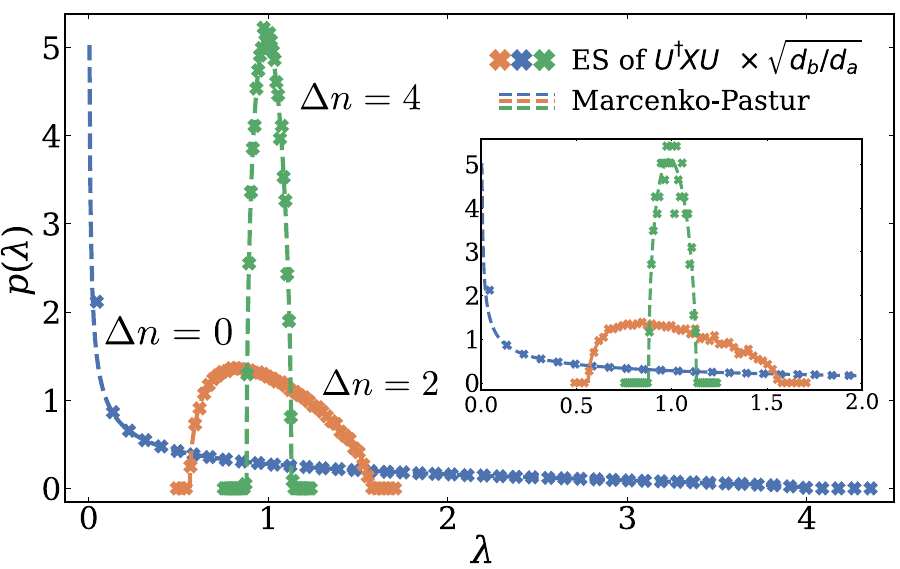
    \caption{Convergence and concentration of unitary-evolved OES to MP for different bipartitions on a system of $N = 10$ sites. $d_A = 2^{N_A}$ and $d_B = 2^{N_B}$ give the subsystem dimensions, and $\Delta N = N_A - N_B$. The scaling factor $\sqrt{d_B/d_A}$ originates in the identical fluctuations of the Pauli coefficients, which controls the average inner product of the columns of the matrix $\alpha_{ij}$ introduced above. The explicit derivation is given in Eq.~\eqref{eq:normalization}. The data represent 128 matrix samples, and the inset is a single sample on a system of size $n = 12$, demonstrating concentration of the spectrum.}
    \label{fig:unitary_convergence_concentration}
\end{figure}

Lastly, due to the time-reversal symmetry of the automaton circuits to be considered later, we also need to understand the OES statistics of random \textit{real} observables.
More precisely, we now consider the OES of the ensemble $\{X_O\equiv O^tXO\}_O$ where $O$ are drawn from the uniform distribution on the orthogonal group.
We model $X_O$ by a matrix $M$ sampled from the GOE, which is just the real part of the GUE. This is equivalent to sampling $M$ from the GUE and setting to zero each coefficient $\alpha_{ij}$ [Eq.~\eqref{eq:alpha}] corresponding to a complex Pauli operator. 
In Appendix~\ref{ap:GOE}, we show that the OES of $M$ can be reduced to that of a direct sum of two real and independent Gaussian matrices. 
We derive an approximate expression for this distribution. 
In a similar spirit to the Wigner surmise, we treat the smallest example that effectively captures the behavior of the larger ensemble, computing the joint OES spacing distribution of a $2\times 2$ GOE block and a single random Gaussian variable. We find this is described by the distribution
\begin{equation}
\label{eq:pgoe}
p_{\text{GOE}}(r) \approx \frac{3}{4}\frac{1+r}{(1+r+r^2)^{3/2}}.
\end{equation}
Note that, like the Wigner surmise, this is not exact for large dimensions.
Unlike the Wigner surmise, our prediction for $p_{\text{GOE}}(r)$ does not vanish at $r=0$, indicating the absence of level repulsion owing to the appearance of independent blocks.
In Fig.~\ref{fig:classical_ensembles}, we compare the OES and level spacing ratios of a GOE matrix (green crosses) with the OES of the ensemble $\{X_O\}_{O}$ (solid green dots). 
We find that the spacing ratios agree well with the prediction (solid green line), and the two ensembles agree almost perfectly with each other.

\subsection{\label{subsec:unitary_evo} Locally generated chaotic unitary dynamics}
In realistic scenarios, dynamics are understood as resulting from local interactions, and quantum circuits provide a structured model for such dynamics. 
Quantum circuits that act on a chain of qubits and are comprised of random 2-local gates are colloquially known as random unitary circuits (RUCs). 
It is reasonable to ask whether the model of chaotically evolved observables as random Hermitian matrices breaks down when the evolution is generated by RUCs instead of all-to-all Haar-random unitary operators. 
To answer this, we consider a one-dimensional chain of $N$ qubits and again take as our initial operator the Pauli-$X$ operator on a single site near the middle of the chain. The difference between the emperically measured OES DOS and level spacings from the predicted distributions at subsequent times is quantified by the Kullback-Leibler (KL) divergence, a standard divergence measure between two probability distributions $P$ and $Q$, which is defined as the difference in expected surprisal when evaluated using the ``true" distribution $P$ vs the ``model" distribution $Q$:
\begin{equation}
D_{\text{KL}}(P \Vert Q) = \sum_{x \in X}P(x) \log(\frac{P(x)}{Q(x)})
\end{equation}
This divergence is defined when $P$ is absolutely continuous with respect to $Q$.
As shown in Fig.~\ref{fig:thermalizing in RUCs}, we find that both the empirical OES DOS (red) and spacing ratio statistics (blue) converge at late times to their universal distributions. The data represent $2^{18}$ eigenvalue samples for each system size taken at circuit depths $l = 1, \dots, 15$. 

Fig.~\ref{fig:thermalizing in RUCs} focuses on the evolution of these distributions in time, as measured by the number of brick-wall layers $l$ in the RUC, and on their scaling with system size $N$.
Such features have been studied previously for the \textit{state} entanglement spectrum under RUCs~\cite{chang2019evolution}, and were found to exhibit universal behavior. 
In particular, the entanglement spectrum statistics undergo a radical change, evolving from sharply peaked about the initial product state, then passing through a scale-invariant power-law distribution, and finally equilibrating at late times to the MP distribution. 
We observe a qualitatively similar behavior in the OES of operators evolved under RUCs. 
For a circuit on $N$ qubits, we see the appearance of a power-law OES distribution at depth $l=N/2$ with a non-universal exponent, and convergence to the MP distribution at late times. 
The qualitative similarity between the dynamics of the OES and the state ES under RUCs is not a trivial consequence of state-operator duality.
If $U$ is Haar-random or an RUC approximating a Haar-random unitary, the vectorized initial operator $X$ evolves as $\opket{UXU^{\dagger}} = U \otimes U^\ast \opket{X}$ (where the tensor product here reflects the incoming/outgoing indices instead of the bipartition.) Since $U \otimes U^\ast$ is not Haar-random (or approximately Haar-random) on $U(\mathcal B(H))$, this implies $\opket{UXU^{\dagger}}$ is not simply a Haar-random (or approximately Haar-random) state on the enlarged Hilbert space $\mathcal H_d^{\otimes 2}$.
There is also a difference in the equilibration of the OES spacing ratios (blue) and the spectrum itself (red). 
The former reaches its late-time value already at depth $l = N/2$, which is the time when the extent of the operator light cone stretches across the entire system. 
Note, however, that our comparison of the OES level spacings to the WD distribution is limited by the fact that the expression [Eq.~\eqref{eq:WD}] is not exact in the $d \to \infty$ limit.

\begin{figure}
    \centering
    \def\svgwidth{0.95\linewidth}
\begingroup%
  \makeatletter%
  \providecommand\color[2][]{%
    \errmessage{(Inkscape) Color is used for the text in Inkscape, but the package 'color.sty' is not loaded}%
    \renewcommand\color[2][]{}%
  }%
  \providecommand\transparent[1]{%
    \errmessage{(Inkscape) Transparency is used (non-zero) for the text in Inkscape, but the package 'transparent.sty' is not loaded}%
    \renewcommand\transparent[1]{}%
  }%
  \providecommand\rotatebox[2]{#2}%
  \newcommand*\fsize{\dimexpr\f@size pt\relax}%
  \newcommand*\lineheight[1]{\fontsize{\fsize}{#1\fsize}\selectfont}%
  \ifx\svgwidth\undefined%
    \setlength{\unitlength}{511.39178467bp}%
    \ifx\svgscale\undefined%
      \relax%
    \else%
      \setlength{\unitlength}{\unitlength * \real{\svgscale}}%
    \fi%
  \else%
    \setlength{\unitlength}{\svgwidth}%
  \fi%
  \global\let\svgwidth\undefined%
  \global\let\svgscale\undefined%
  \makeatother%
  \begin{picture}(1,0.85134425)%
    \lineheight{1}%
    \setlength\tabcolsep{0pt}%
    \put(0,0){\includegraphics[width=\unitlength,page=1]{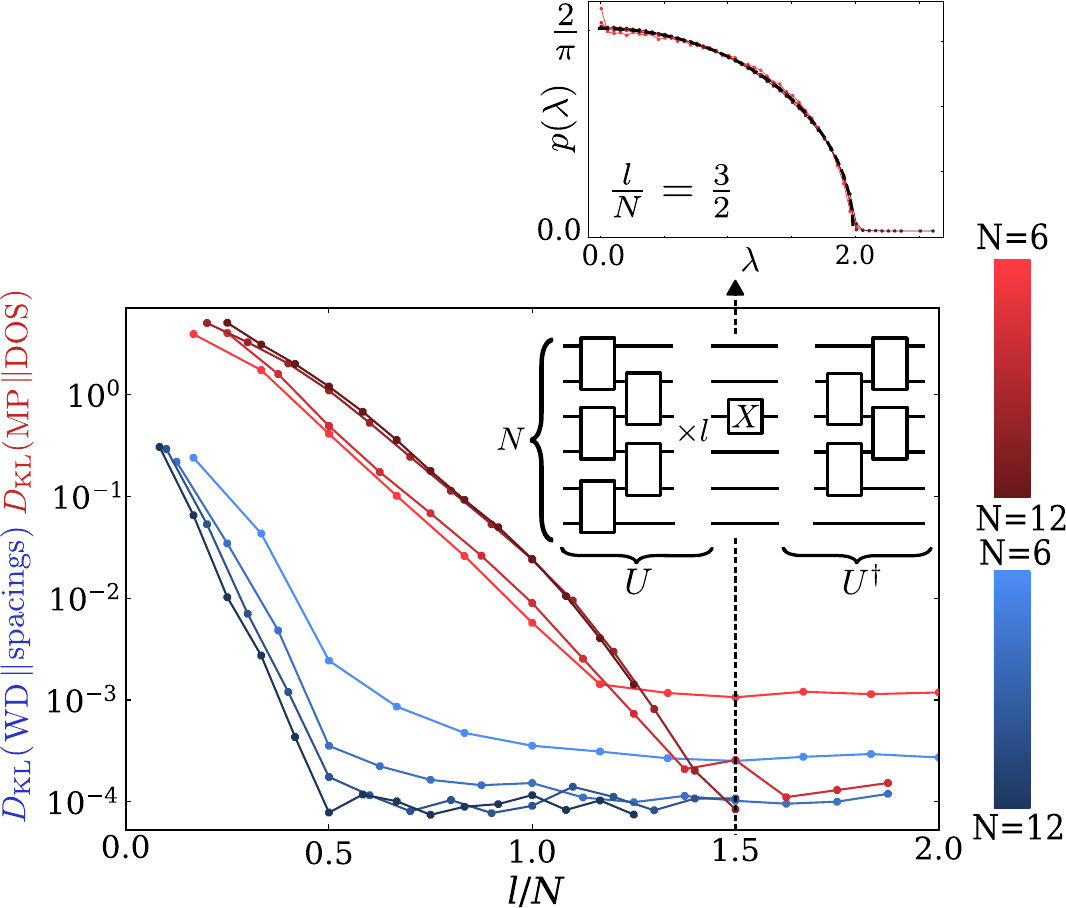}}%
    \put(0.63983071,0.42527765){\color[rgb]{0,0,0}\makebox(0,0)[lt]{\lineheight{0.2}\smash{\begin{tabular}[t]{l}...\end{tabular}}}}%
    \put(0.73372546,0.42572718){\color[rgb]{0,0,0}\makebox(0,0)[lt]{\lineheight{1.25}\smash{\begin{tabular}[t]{l}...\end{tabular}}}}%
    \put(0,0){\includegraphics[width=\unitlength,page=2]{figures/distribution_evolution_unitary_svg-tex.pdf}}%
  \end{picture}%
\endgroup%

    \caption{Convergence of DOS to MP and level spacings to WD as a function of circuit depth. The inset in the lower panel shows a circuit diagram of a local Pauli-$X$ gate evolved under a random unitary circuit on $N$ qubits with $l$ layers. The top left panel shows that the OES passes through a power-law with an apparently non-universal exponent $\eta$ at $l/N = \frac{1}{2}$, which is when the lightcone reaches the edge of the system. At the same time, the level spacings are nearly indistinguishable from WD. \edit{While the level spacings appear not to asymptotically approach the WD distribution, this is because} the WD distribution is not exact for large system sizes. \edit{We plot with the gray dashed line the minimum that can be reached numerically, which we obtain from directly sampling GUE matrices of the same system size and computing $D_{\text{KL}}$.} The top right inset shows the clear emergence of the semicircle in the OES at late times.}
    \label{fig:thermalizing in RUCs}
\end{figure}

\section{\label{sec:sec3}OES in Automaton Dynamics}

Automaton circuits share many of the features of chaotic quantum dynamics~\cite{bertini2024quantum, iaconis2019anomalous}, including the generation of extensive operator entanglement. 
In this section, we show that the OES of automaton evolutions manifests a distinct universality class from that of chaotic unitary dynamics.

\subsection{Automaton OES and Bernoulli matrices}
\subsubsection{OES of a random permutation}
Automaton circuits are a subset of unitary quantum circuits whose action maps computational basis states to computational basis states.
As an example, a CNOT gate maps $\ket{11}\leftrightarrow\ket{10}$ while leaving $\ket{00}$ and $\ket{01}$ fixed; it generates no entanglement when acting on a computational basis state, but does when acting on a superposition of computational basis states. In general, an automaton circuit $P$ is typically defined with the possibility of applying a phase: with respect to the computational basis $\ket{n}$, $P\ket{n} = \e^{\ii\theta_n}\ket{P(n)}$, where $P(n)$ is a permutation applied to $n$. We focus on the case $\theta_n = 0$ for the connection to classical binary logic. Much like a generic unitary evolution can be modeled by a Haar-random unitary matrix $U$ acting on $\mathcal H_d$, a generic automaton evolution can be modeled by a random permutation matrix $P$ acting on the same Hilbert space.

For a Haar-random unitary $U$,
we argued in Sec.~\ref{subsec:classical_ensembles} that the partial transpose $t_{AB}$ effectively removes correlations between matrix elements of $U$, such that $U^{t_{AB}}$ resembles a sample from a classical random matrix ensemble in the $d\to\infty$ limit.
\edit{While identifying general criteria for the emergence of independence through the partial transpose in this way is an interesting mathematical question which we leave for future work, the feature that appears to connect these two ensembles is invariance under the unitary and permutation group respectively.}
Applying this logic to the ensemble of random automata $\{P \in S_{2^N}\}$ (where $S_{2^N}$ is the permutation group on the $N$-qubit computational basis) suggests that the relevant random matrix ensemble describing $P^{t_{AB}}$ is one of matrices with entries randomly drawn from the set $\{0,1\}$.
A square matrix having entries $0$ or $1$ drawn with probability $p(1)=1-p(0)=p$ is called a Bernoulli random matrix~\cite{guionnet2021bernoulli}.
After taking the partial transpose, we expect the matrix $P^{t_{AB}}$ to have a single $1$ per column, on average, in the $d\to\infty$ limit.
This leads us to the hypothesis that the OES statistics of $P$ should be compared to the statistics of the singular values of Bernoulli matrices $B$---or equivalently, the eigenvalues of $BB^t$---with $p=1/d$ in the limit $d\to\infty$.

In Appendix~\ref{app:moments}, we show that calculating the $k^{\text{th}}$ moment of the averaged eigenvalue distribution of $BB^t$ can be reduced to a partition counting problem. 
Performing this numerically for the first few moments gives
\begin{align}
\label{eq:mk}
m_k &= \frac{1}{d}\mathbb E_B \tr[(BB^t)^k] \notag\\
&= 1, 3, 12, 57, 303, 1747, 10727,  69331 \dots
\end{align}
This sequence was found previously in Ref.~\cite{khorunzhy2000asymptotic} as the even moments of the spectrum of the adjacency matrix of a random undirected graph. Ref.~\cite{khorunzhy2000asymptotic} obtains a recursion relation for $m_k$ and shows that $(k/4)^{k/4} \lesssim m_k \lesssim (Ck)^{k}$ for some constant $C$, which implies that this is the unique distribution with this sequence of moments and also that it has unbounded support. 

We can also show that these moments concentrate in the large-dimension limit, in the sense that the fluctuations go to zero as $d \to \infty$:
\begin{align}
P\qty(\left|\frac{1}{d}\tr[(BB^t)^k] - m_k\right| > \epsilon )\to 0
\end{align}
for any $\epsilon > 0$. 
The proof is also found in Appendix~\ref{app:moments}. 

The same argument suggests that an automaton circuit with phases should be modeled by a related matrix ensemble $\{B_{\theta}\}_{B, \theta}$, where $(B_\theta)_{mn} = B_{mn}\e^{\ii \theta_{mn}}$ with $B$ a random Bernoulli matrix and $\{\theta_{mn}\}_{mn}$ i.i.d. around the circle. It turns out that the singular value distribution of this ensemble is the same as that of the Bernoulli ensemble regardless of the distribution of $\theta_{mn}$, so we expect that this universal entanglement spectrum for automaton circuits holds regardless of the phases. This prediction agrees with numerics. The proof is given in Appendix~\ref{app:generic}.

\begin{figure}
\centering
\def\svgwidth{.95\linewidth}
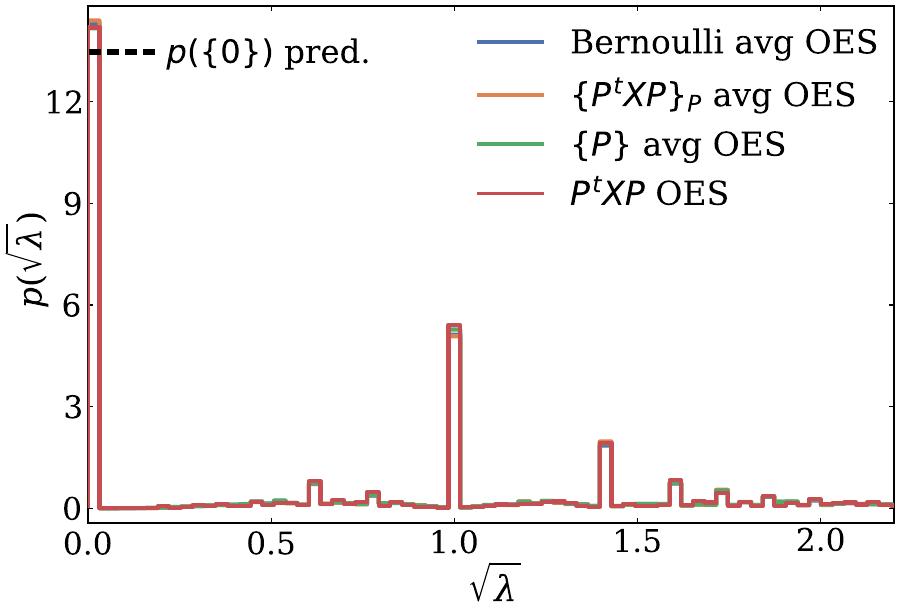
\caption{The average OES of the Bernoulli matrix ensemble and the average OES of the ensembles $\{P\}_P$, where $P$ is a random automaton, and $\{P^tXP\}_P$, where $X$ is a Pauli-$X$, are seen to agree. The OES of $P^tXP$ for a \textit{single} automaton sample $P$ also shows good agreement, suggesting that this is not just an average phenomenon. This defines another universal spectrum which is sharply distinct from the semicircle law for random unitary matrices. The prediction for the average entanglement rank is $d(1-p(\{0\}))$, where $p(\{0\}) \gtrsim \e^{-\frac{1}{\e}}+\frac{2}{\e} -1 \approx 0.43$. This prediction is made by using the independence assumption to estimate the average number of unique columns in a sample and is shown by the dashed line (rescaled by the bin size). The numerical value is $p(\{0\}) \approx 0.46$
\label{fig:automaton_ES}
}
\end{figure}

In Fig.~\ref{fig:automaton_ES}, we plot the average OES DOS of a single random permutation on $N=12$ qubits against the average spectrum of a Bernoulli matrix, the average OES of random automata, and the average OES of the ensemble $\{P^tXP\}_P$, where $X$ is a Pauli-$X$. These distributions all agree remarkably well. Furthermore, we see a sharp qualitative difference with the semicircle distribution governing the OES of random unitary dynamics. In particular, accumulations of levels (also known as ``atoms") are seen at certain algebraic integers, most prominently $\lambda=0$ and $\lambda=1$. 

The agreement between the automaton and Bernoulli OES can be a useful calculational tool. The value of $p(\lambda = 0)$ is $1-r/d_B^2$, where $r$ is the average Schmidt rank of a permutation. We can lower-bound this value by the average number of repeated columns in $B$. For the case $d_A = d_B$, we derive the approximate expression $p(\{0\}) \approx \e^{-\frac{1}{\e}}+ \frac{2}{\e} - 1$ in Appendix~\ref{sec:kernelprob} by estimating the number of linearly independent columns in a Bernoulli matrix of the same size. 
The emergence of other atoms in the singular value distribution of Bernoulli matrices $B$ can be understood by viewing $B$ as the adjacency matrix of a directed graph. The adjacency matrix has the property that it can be block-diagonalized into blocks corresponding to the connected components of the graph. 
In Appendix~\ref{app:moments}, we provide numerical evidence that the graph topology is controlled by that of a random undirected graph with edge probability $pd = 2$. In particular, in the large-dimension limit, there is one large component of the graph which increases with system size and is responsible for any continuous component of the spectrum. The remaining components are trees of size $\mathcal O(1)$, and the generation of these components is predicted by a Poisson branching process \cite{alon2015probabilistic}. Eigenvalues resulting from these trees are algebraic integers of low degree. The comparison between this prediction and the topology of the graph corresponding to $(P^tXP)^{t_{AB}}$ is shown in Appendix~\ref{app:moments} in Fig.~\ref{fig:cmp_sizes}. However, removal of these components does not appreciably change the singular component of the spectrum, suggesting that eigenvector localization on subtrees of the large component is also responsible for the point spectrum.

\subsubsection{OES under automaton dynamics}

Under unitary evolution, the OES DOS of a local observable $X$ follows the MP distribution, as shown previously. 
Automaton evolution maps diagonal operators to diagonal operators, and for this reason, most of the features of automaton circuits that resemble random unitary dynamics are found in off-diagonal operators \cite{bertini2024quantum}. 

For completeness, we will first treat diagonal observables. If $Z$ is a local Pauli-Z operator and $P$ is a random automaton, then we can model $P^tZP$ as a diagonal operator with $-1,1$ randomly on the diagonal with equal probability. The partial transpose $t_{AB}$ takes the diagonal indices, which are of the form $(i, j),(i, j)$ with $i \in \{1 \dots d_B\}$ and $j \in \{1 \dots d_A\}$ to the indices $(i,i),(j,j)$. From this it is seen that, of the $d_A^2$ columns of the matrix $(P^tZP)^{t_{AB}}$ each of length $d_B^2$, precisely $d_A$ of them have non-zero elements corresponding to the choices of $j$. Each of these columns has $d_B$ non-zero entries corresponding to the choices of $i$. Thus if $r$ is a nonzero column of $(P^tZP)^{t_{AB}}$, then $\Vert r \Vert^2 = 2^{N/2}$. This suggests that the non-zero values in the OES of $2^{-N/4}\times P^tZP$ are MP-distributed, as confirmed in Fig~\ref{fig:generic_operator}. If $d_A = d_B$, then the probability that a given Schmidt coefficient is zero, $p(\{0\}) = 1-d^{-1/2}$, goes to unity as $d_B \to \infty$. 

\begin{figure}
    \centering
    \def\svgwidth{0.95\linewidth}
    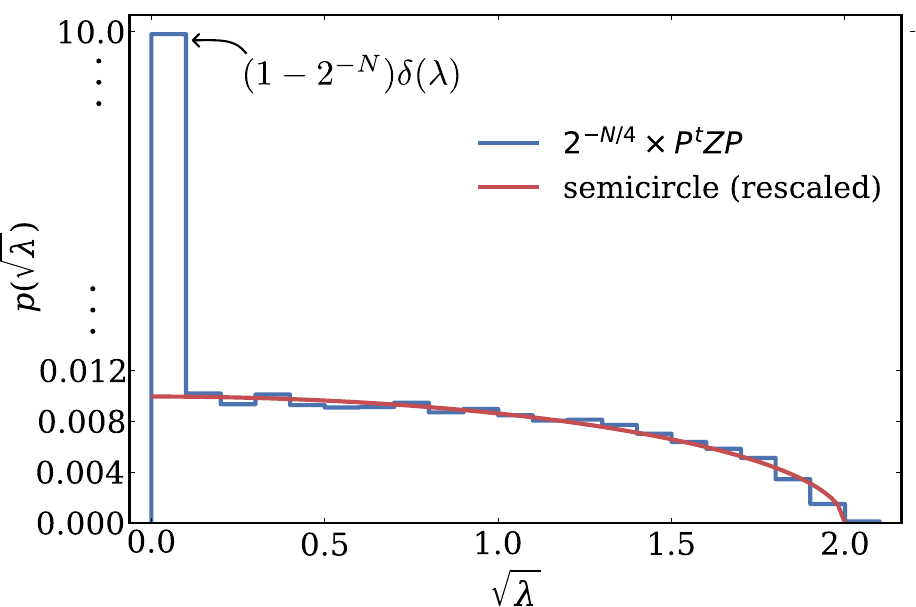
    \caption{A Pauli-$Z$ operator is evolved under a random automaton $P$. The data show $64$ samples on $N=12$ qubits, totaling $2^{18}$ singular values. The $2^{12}$ non-zero singular values a follow the MP distribution upon rescaling as explained in the text. The OES DOS of $2^{N/4}P^tZP$ then becomes $p(\lambda) = (1-2^{-N})\delta(\lambda) + 2^{-N}\text{MP}(\lambda)$, which agrees with the figure.}
    \label{fig:generic_operator}
\end{figure}

We now turn our focus to off-diagonal operators. If $X$ is a Pauli-$X$ operator and $P$ is a random permutation, then the permutation $P^tXP$ will also appear random, so we expect the OES of $P^tXP$ to be determined as well by the singular value distribution of the ensemble of Bernoulli matrices. This prediction shows good agreement with the numerics, as presented in Fig.~\ref{fig:automaton_ES}. In general, a single-qubit off-diagonal observable $A$ must take the form $A \propto \smqty(0 & z \\ z^\ast & 0)$ for some complex $z$ with $|z| = 1$, and therefore any product of single-qubit off-diagonal observables has the same structure as $P^tXP$, but with roots of unity instead of 1's in the matrix elements. As argued in the previous section, this ensemble is modeled by $\{B_{\theta}\}_{B, \theta}$, which has the same singular value distribution as the Bernoulli ensemble. 

Numerically, we compare the moments of the OES of $P^tXP$ to the sequence of moments given by Eq.~\eqref{eq:mk} in Fig.~\ref{fig:bernoulli_vs_permutation}, and find that the estimated moments approach their predicted values exponentially in the system size. Since the method of moments is applicable here, this suggests that the OES of $P^tXP$ concentrates to the limiting distribution of the singular values of Bernoulli matrices. Similarly, the OES of $U^\dagger X U$ concentrates to the limiting distribution of a GUE matrix. Despite their apparent similarities, the OES sharply distinguishes random unitary and automaton dynamics.

\begin{figure}
    \centering
    \def\svgwidth{0.95\linewidth}
    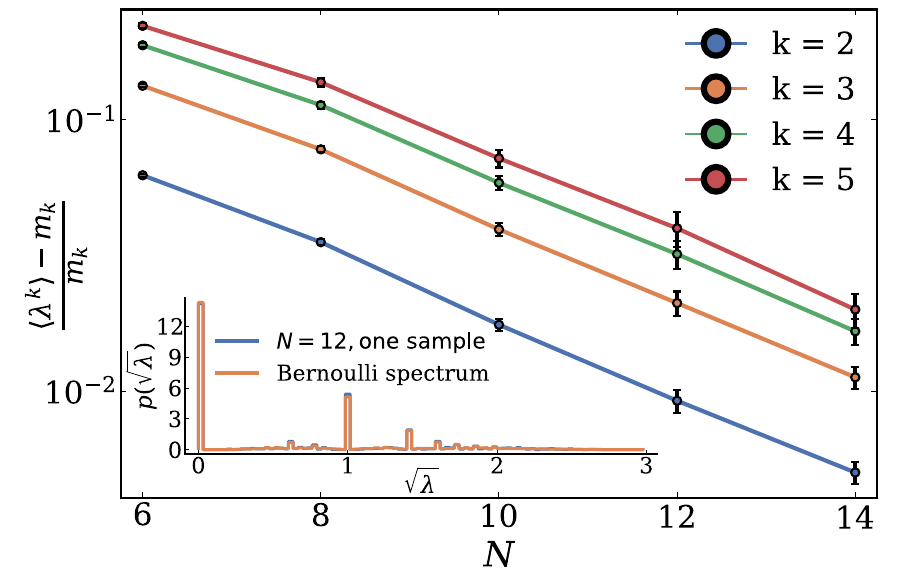
    \caption{Convergence of the OES moments of $P^tXP$, where $X$ is a local Pauli-$X$ and $P$ is a random automaton, to the moments of the Bernoulli ensemble with increasing system size. 
    The first few moments of the Bernoulli spectral distribution $m_k$ are computed using the formula in Eq.~\eqref{eq:mk}. The data represent $2^{20}$ values $\{\lambda\}$ from the OES of randomly sampled automaton circuits at each system size.  The error bars are estimated using resampling. The inset shows the estimated OES DOS of a single automaton at system size $N=12$, which agrees well with the distribution shown in Fig.~\ref{fig:automaton_ES}, providing evidence that the moments are indeed a good measure of convergence in distribution.}
    \label{fig:bernoulli_vs_permutation}
\end{figure}

\subsection{Locally generated random automaton evolutions}
Since evolution of off-diagonal operators under local automaton circuits is known to resemble RUCs from the perspective of, e.g., hydrodynamics or certain OTOCs, we expect a qualitative similarity in the way the OES approaches the universal distribution.
In Fig.~\ref{fig:local_automaton_ES}, we study the OES of a local Pauli-$X$ operator in the middle of a chain of $N$ qubits evolved under a local random automaton circuit as a function of circuit depth and system size.

We sample local automaton circuits by composing three layers of permutations each acting on three qubits (see inset at lower-right), as it has been shown that 3-local permutations can generate the entire alternating group, consisting of even permutations (odd permutations require an $N$-qubit gate or an ancilla qubit) \cite{coppersmith1975generators}. 
Since the OES distribution of a random automaton-evolved observable is expected to have a dense set of atoms and a continuous component, the KL divergence will not be an accurate measure of the similarity between the universal distribution and the ideal distribution. 
In contrast, the moments $m_k$ (see Eq.~\eqref{eq:mk}) are easy to compute.
We therefore track the convergence of the 2nd moment of the OES of the evolved operator, which we denote by $\expval{\lambda^2}$, to the ideal value $m_2=3$ as a function of the number of layers $l$ in the circuit. 
We observe a behavior which is qualitatively similar to RUCs: the relative error $\Delta m_2 = \frac{\langle \lambda^2\rangle - m_2}{m_2}$ converges as a function of $l/N$ exponentially with a constant independent of system size, and then at late times thermalizes to a value that decreases exponentially with system size.
This provides a more precise description of the way operator dynamics in random automaton circuits resemble those in RUCs with respect to two different universal distributions.

\begin{figure}
    \centering
    \def\svgwidth{0.95\linewidth}
\begingroup%
  \makeatletter%
  \providecommand\color[2][]{%
    \errmessage{(Inkscape) Color is used for the text in Inkscape, but the package 'color.sty' is not loaded}%
    \renewcommand\color[2][]{}%
  }%
  \providecommand\transparent[1]{%
    \errmessage{(Inkscape) Transparency is used (non-zero) for the text in Inkscape, but the package 'transparent.sty' is not loaded}%
    \renewcommand\transparent[1]{}%
  }%
  \providecommand\rotatebox[2]{#2}%
  \newcommand*\fsize{\dimexpr\f@size pt\relax}%
  \newcommand*\lineheight[1]{\fontsize{\fsize}{#1\fsize}\selectfont}%
  \ifx\svgwidth\undefined%
    \setlength{\unitlength}{443.59869385bp}%
    \ifx\svgscale\undefined%
      \relax%
    \else%
      \setlength{\unitlength}{\unitlength * \real{\svgscale}}%
    \fi%
  \else%
    \setlength{\unitlength}{\svgwidth}%
  \fi%
  \global\let\svgwidth\undefined%
  \global\let\svgscale\undefined%
  \makeatother%
  \begin{picture}(1,0.64923546)%
    \lineheight{1}%
    \setlength\tabcolsep{0pt}%
    \put(0,0){\includegraphics[width=\unitlength,page=1]{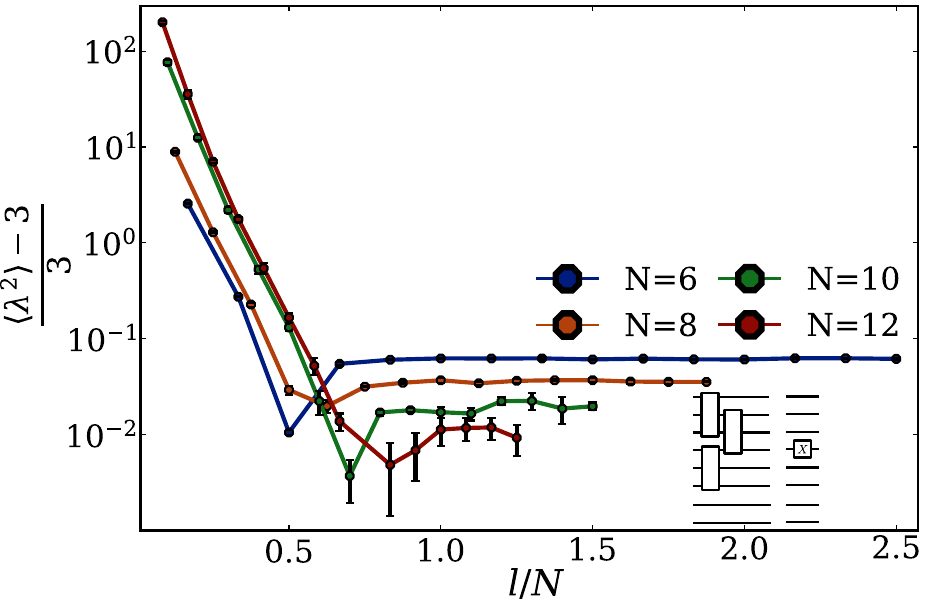}}%
    \put(0.83468488,0.14201828){\color[rgb]{0,0,0}\makebox(0,0)[lt]{\lineheight{1.25}\smash{\begin{tabular}[t]{l}...\end{tabular}}}}%
    \put(0.88898177,0.14228697){\color[rgb]{0,0,0}\makebox(0,0)[lt]{\lineheight{1.25}\smash{\begin{tabular}[t]{l}...\end{tabular}}}}%
    \put(0,0){\includegraphics[width=\unitlength,page=2]{figures/local_permutation_moments_svg-tex.pdf}}%
  \end{picture}%
\endgroup%

    \caption{Convergence of second moment of automaton-evolved local $X$ OES with system size and circuit depth. The construction of the circuits using 3-local permutations, similar to the construction of a RUC shown in Fig.~\ref{fig:thermalizing in RUCs}, is depicted in the lower-right-hand corner. This convergence exhibits a similar scaling as the KL divergence of local unitary circuits from MP, exhibiting an analogue between classical and quantum information chaos. The inset shows a comparison of the full OES of a $2^{12} \times 2^{12}$ Bernoulli matrix to an automaton-evolved $X$ operator at system size $N=12$ and circuit depth $l=15$, supporting the idea that the second moment is monitoring the convergence of the full spectrum.}
    \label{fig:local_automaton_ES}
\end{figure}

\section{\label{sec:sec4}Classical-Quantum Crossover in Superposition-Doped Automaton Circuits}
It is well-known~\cite{aharonov2003simple} that permutation gates together with the Hamadard gate form a universal gate set; classical logic, when augmented with superpositions, is enough to approximate any unitary operation. In this section, we show that this connection between universal classical and quantum gates is reflected by a transition in the OES from the Bernoulli OES identified in the previous section to the GUE OES identified in Sec.~\ref{sec:Haar OES}. This transition is probed by ``doping" random automaton circuits with quantum gates that create superpositions (referred to in the Introduction as SG gates). We examine the effect of doping with SG gates on the OES DOS and also on the level spacing ratios, finding that the OES DOS becomes exponentially close to the MP distribution in the number of such gates independent of system size, and the level spacing ratios become exponentially close to the WD distribution in system size for any number of superposition gates.

\subsection{OES Distribution}
In Ref.~\cite{zhou2020single}, it was shown that the ES spacing ratios of a random product state evolved under a Clifford circuit undergo a transition from Poisson to Wigner-Dyson statistics as the circuit is ``doped" with a density of $T-$gates that is vanishingly small with increasing system size. 
This was conjectured to signal the transition to a unitary 4-design, which was later proven rigorously in Ref.~\cite{haferkamp2020quantum}.
The corresponding increase in simulation complexity was studied using the OEE growth in Ref.~\cite{dowling2025bridging}.
Clifford circuits are classically easy to simulate because they act as automaton circuits on the space of operators. 
If $C$ is a Clifford operator acting by conjugation on a Pauli operator $P_i$, then for each $1 \leq i \leq d^2$, we have $C^\dagger P_i C = \e^{\ii\theta_{i}}P_j$ for some $j$ by the fact that the Clifford group is the normalizer of the Pauli group. 
If $\mathcal C$ is the adjoint representation of $C$ acting on Pauli-string space, then it is equivalent to say $\mathcal C\opket{P_i} = \e^{\ii\theta_i}\opket{P_j}$, which is the definition of an automaton. 
Similarly, with $T = \e^{-\ii\frac{\pi}{8}Z}$, for the $2^n$ Pauli operators $P_i$ which anti-commute with $Z$, we have $TP_iT^\dagger = \frac{1}{\sqrt{2}}P_i(1+\ii Z)$. 
Writing this situation in Pauli-string space, $\mathcal T \opket{P_i} = \frac{1}{\sqrt{2}}(\opket{P_i} + \ii\opket{P_iZ})$. 
In other words, the complex state dynamics result from the superposition of just a few trajectories in Pauli-string space.

Automaton circuits act on state space equivalently to the way Clifford circuits act on Pauli-string space. 
The addition of SG gates to an automaton circuit results in state dynamics consisting of a superposition of automaton dynamics along multiple trajectories. 
We show in this section that doping an automaton circuit with SG gates provokes a transition in the OES in a manner very similar to the one induced by $T$ gates in Clifford circuits. 
In Fig.~\ref{fig:DOS_convergence}, we observe that the addition of Hadamard gates causes the OES DOS of initially local $X$ and $Z$ operators to converge to the MP distribution. 
These data were obtained from $2^{20}$ eigenvalues at each system size by conjugating a local Pauli-$X$ (top panel) or Pauli-$Z$ operator (bottom panel) with a random Hadamard-doped automaton circuit $U$. This circuit is obtained by sampling two random permutations $P_1, P_2$ and $n_H$ locations to place Hadamard gates $H$ in between these permutations, obtaining $U = P_1 H^{\otimes n_H}P_2$.
The KL divergence from the MP distribution is initially independent of the system size and then saturates at a value that is exponentially small in system size. 

\begin{figure*}
    \centering
    \def\svgwidth{0.45\linewidth}
    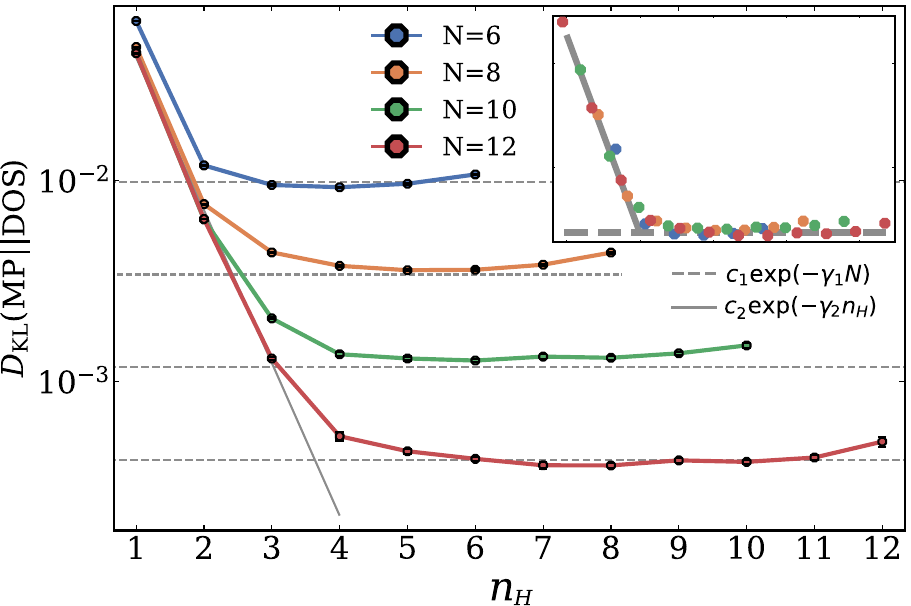
    \def\svgwidth{0.45\linewidth}
    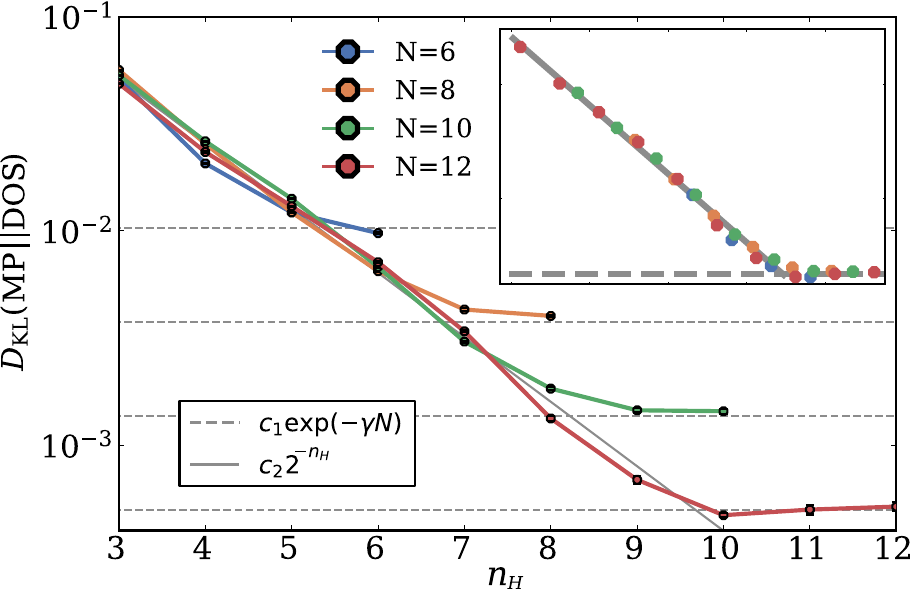
    \caption{As the number of Hadamard gates $n_H$ increases, the entanglement spectrum of an evolved local operator approaches the universal MP distribution. The top panel shows an initial $X$ operator and the bottom panel shows an initial $Z$ operator. For small $n_H$, the convergence is exponential in $n_H$, and for $n_H \sim N$, the convergence is exponential in $N$. This is illustrated in the inset by a universal scaling function. The data represent $2^{20}$ eigenvalue samples at each system size, and the error bars (too small to be visually apparent) are estimated via resampling.}
    \label{fig:DOS_convergence}
\end{figure*}

The dependence on whether the initial operator is diagonal or off-diagonal disappears when  $n_H \sim N$. 
The rate of convergence to the MP distribution is dependent on the initial evolved operator: comparing the top and bottom panels of Fig.~\ref{fig:DOS_convergence}, the convergence for an initial Pauli-$Z$ operator is significantly slower. 
This could be explained heuristically through the localization in Pauli-string space. 
If $Z$ is conjugated by $U$, then first conjugating by $P_1$ creates a random superposition of tensor products of local $Z$ operators. Each string of $Z$-operators is then conjugated by the randomly placed Hadamard gates, turning some of them into $X$ operators, which connects exponentially larger blocks in Pauli-string space. Since the convergence to the MP distribution shown in the top panel Fig.~\ref{fig:DOS_convergence} is much faster than $2^{-n_H}$, eigenvalues from the diagonal blocks dominate the KL divergence, leading to the $D_{\text{KL}}(\text{MP} \Vert \text{DOS}) \sim 2^{-n_H}$ scaling observed in the lower panel of Fig.~\ref{fig:DOS_convergence}. This argument is given more precisely in Appendix \ref{app:operator_dependence}.

\subsection{Level Spacings}
We examine the level-spacing statistics of the OES under superposition-doped automaton circuits, which provide more information about the operator than the DOS alone. 
In particular, the MP distribution can arise whenever an operator in the Pauli basis has blocks that are extensively large in the system size and fluctuate independently.
Since the level spacings are sensitive to the joint distribution of the eigenvalues, these spacings are another way to probe ergodicity in Pauli-string space. Since $X_P$ and $Z_P$, where $P$ is a random automaton, are completely real, we expect their level spacings to be reminiscent of GOE operators rather than GUE or Haar-random operators. 
This clearly must hold for a real operator, and in fact it holds for any local Clifford operator, as demonstrated in Appendix~\ref{ap:GOE}, see Fig.~\ref{fig:spacingstats}.
Since we only have an approximate expression for $p_{\text{GOE}}$ [Eq.~\eqref{eq:pgoe}], to examine the convergence of the level spacings in superposition-doped automata, we opt to replace the $H$ gates with the single-qubit gate $R_x=\e^{-\ii\frac{\pi}{4}X}$, which re-introduces level repulsion.
After making this substitution, we find that the convergence of the OES level spacings to Wigner-Dyson is qualitatively different than the convergence of the DOS to MP, as shown in Fig.~\ref{fig:level_spacing_convergence}. Principally, unlike the KL divergence of the DOS, the KL divergence of the level spacings decreases exponentially with system size for small $n_{R_x}$.
This indicates that the addition of a single superposition-creating gate in automaton circuits acts similarly to the $T$ gate in Clifford circuits: a homeopathic dose in a large system drives a sharp transition in the OES of an evolved local operator in the same way that a single $T$ gate in a Clifford circuit drives a transition in the level spacings of an evolved random product state.

\begin{figure}
    \centering
    \def\svgwidth{0.95\linewidth}
    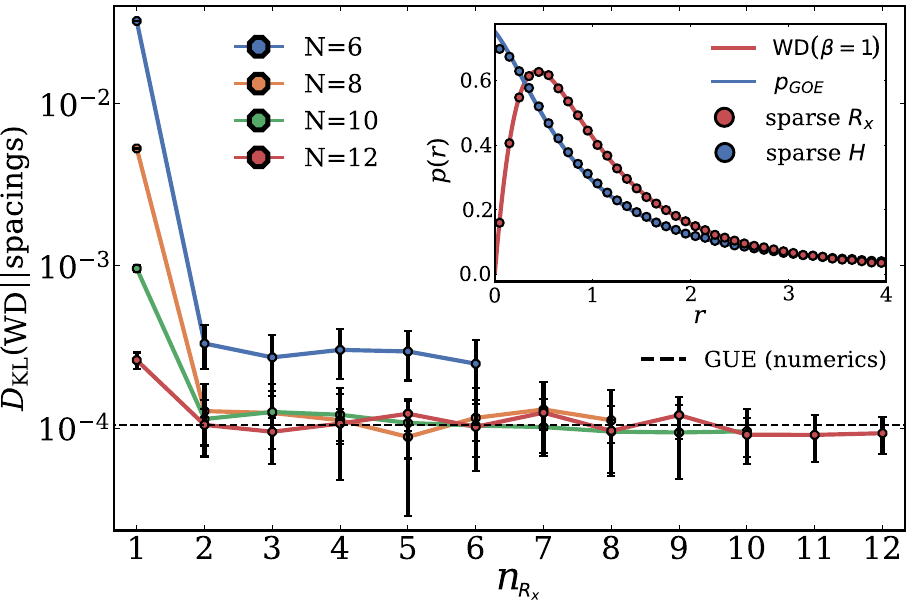
    \def\svgwidth{0.95\linewidth}
    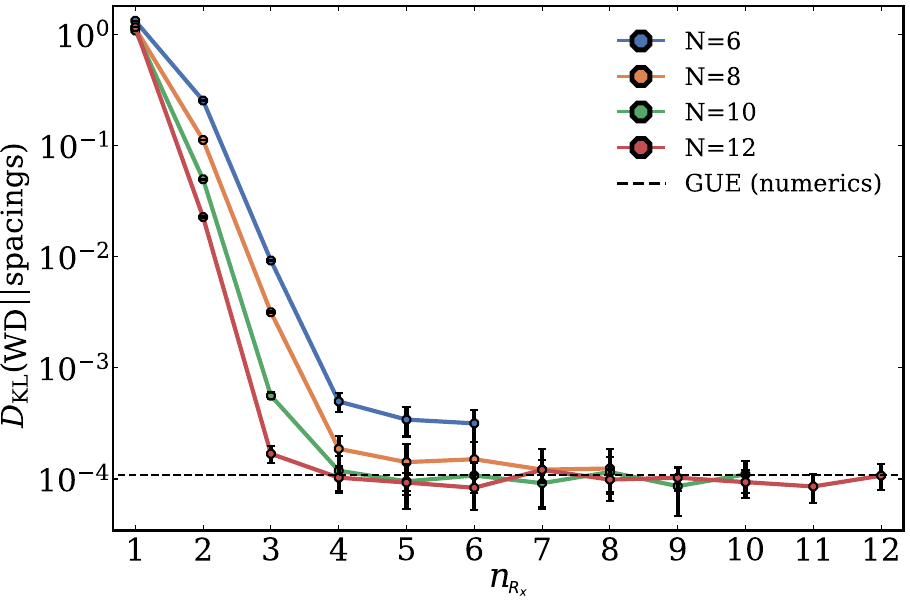
    \caption{
    Convergence of level spacing ratios to WD with number of $R_x$ gates inserted into the circuit. The left panel shows an initial operator $X$ and the right panel shows initial operator $Z$, reflecting that this convergence holds for both diagonal and off-diagonal operators. 
    The convergence is exponential in system size at all $n_{R_x}$, signaling that a single non-classical gate produces a drastic change in the operator dynamics. The inset in the top panel shows the spacing ratio distribution at $N = 12, n_{H,R_x} = 6$ for $R_x$-doped vs Hadamard-doped circuits, where the former are WD with $\beta = 1$ and the latter follow the distribution $p_{\text{GOE}}$ for a random real observable. The dashed horizontal line is the KL divergence obtained numerically with a comparable sample of GUE matrices.}
    \label{fig:level_spacing_convergence}
\end{figure}

One immediate application of this finding is to the simulation of chaotic operator dynamics. Just a finite number of SG gates in an automaton circuit causes the dynamics to resemble those of random unitary circuits significantly more than just an automaton circuit alone. A finite number of SG gates only increases the state-space simulation complexity by a finite prefactor, using e.g. quantum Monte-Carlo sampling, so this could be a route to computationally tractable simulations of chaos-like operator dynamics in large systems.

\section{\label{sec:conclusion} Conclusion}
This work advances our understanding of operator entanglement dynamics, distinguishing between automaton circuits and fully quantum dynamics through the lens of the operator entanglement spectrum (OES). We began by showing that the OES takes on a universal form in random unitary dynamics, both under Haar random unitary operators and thermalizing under random unitary circuits. We then examined automaton circuits, which have been studied in the past for producing operator evolution similar to that of Haar-random unitary operators, and showed that the OES falls into a different universality class. This clearly establishes that the full OES carries information about quantum dynamics that cannot be deduced from standard measures. We examined some of the properties of this distribution, including the moments, concentration of the measure, and emergence of a large singular component connected to random graph topology. The OES of random automaton circuits may have applications to classical cryptography due to the importance of pseudorandom permutations. Several recent studies have connected quantum properties of permutations such as the entanglement entropy \cite{bakker2024operator, pizzi2022bridging} and properties of quantum chaos such as the vanishing of OTOCs \cite{chamon2022quantum} to classical properties of binary gates. Lastly, we introduce superposition-doped automaton circuits. We find that finitely many Hadamard gates in an automaton circuit are sufficient to drive the OES exponentially close to its thermal behavior, and a vanishing density of $R_x$ gates are sufficient to create level repulsion and drive a transition to Wigner-Dyson statistics. This creates a bridge between classical automaton dynamics and random quantum dynamics, provides a theoretically interesting model in which to study operator dynamics, and suggests a method for efficiently simulating operator dynamics that are more generic than automaton dynamics alone. 

\edit{One might wonder if the appearance of the semicircle distribution signals the emergence of a unitary or orthogonal design in anology with \cite{haferkamp2020quantum}. 
Because the OES is not sensitive to the microscopic distribution of the matrix elements, the semicircle does not in general imply the convergence to a design.
However, it might be interesting to investigate whether superposition-doped automaton circuits provide a way of lifting cryptographically secure classical permutations to pseudo-random unitaries similarly to the PFC ensemble in Ref.~\cite{PFCs} or exponentiating random permutations in Ref.~\cite{Chen_2024}.}

Since the MP distribution is a direct way to quantify delocalization in operator space, we intend to explore further applications to quantum cryptography and quantum complexity theory in future work.

\begin{acknowledgments}
B.T.M.~was supported by the U.S.~Department of Energy Office of Science, Science Undergraduate Laboratory Internships (SULI) program under its contract with Iowa State University, Contract No.~DE-AC02-07CH11358.
This work was supported in part by the National Science Foundation under grants DMR-2143635 and DMR-2611305 (T.I.), DMR-2238895 (J.H.W.), and GCR-2428487 (C.C.).
The code and data used in this work are made available in Ref.~\cite{mcdonough2025automata}.
\end{acknowledgments}

\bibliographystyle{quantum}
\bibliography{refs}  

\newpage
\appendix
\onecolumngrid

\section{Limiting OES in invariant ensembles}
The common thread connecting the ensembles of automaton- and unitary-evolved observables is invariance under conjugation by the compact groups $\SU(d)$, $\SO(d)$, or $S_{d}$ (the permutation group). In this section, we explore the consequences of the invariance of operator ensembles under conjugation by $\SU(d)$ on the limiting OES. First, the MP distribution with parameter $\gamma$ [Eq.~\eqref{eq:MP}] can be understood as the limiting distribution of the eigenvalues of a matrix $\sum_{i=1}^d \ketbra{r_i}{r_i}$, where $\{\ket{r_i}\}_{i=1}^d$ are statistically independent and distributed around zero with $\overline{\braket{r_i}{r_i}} \to \gamma$ as $d \to \infty$. To apply this, consider an operator $M$ and the matrix $\alpha_{ij} \equiv \opbraket{P_i \otimes P_j}{M}$. The inner product of two columns of this matrix can be written as
\begin{align}
\sum_{k}\alpha_{ik}\alpha_{jk}^\ast &= \sum_{k} \opbraket{P_i \otimes P_k}{M}\opbraket{M}{P_j \otimes P_k}=\opbra{P_i}\left( \tr_{\mathcal B(H_B)}\opketbra{M}{M}\right) \opket{P_j}
\end{align}
making it clear that the requirement for the emergence of the MP distribution, $\sum_{k}\alpha_{ik}\alpha_{jk}^\ast \to \frac{\gamma}{\mathcal N}\delta_{ij}$ as $d \to \infty$, also implies that $\tr_{\mathcal B(H_B)}\opketbra{M}{M} \to \frac{\gamma}{\mathcal N}\mathds 1_{\mathcal B(H_B)}$, where $\mathcal N$ is a normalization factor. 
This should be compared to a similar condition for the reduced density matrix of a Haar-random state. 

Now for simplicity, assume $M$ is traceless. 
To compare the entanglement spectrum of ensembles of different dimension, we derive the normalization factor $\mathcal N$ for the ensemble $\{U^\dagger MU\}_U$, where $U$ is Haar-random, to have an MP-distributed entanglement spectrum. 
Since for any nonidentity $P_i, P_j \in \mathcal P_A$ we may choose a unitary $U_{ij}$ such that $U_{ij}^\dagger P_i U_{ij} = P_j$, we get
\begin{align}
\mathbb E_U\qty[\sum_{k} \tr(P_i \otimes P_k U^\dagger M U) \tr(U^\dagger M^\dagger U P_i \otimes P_k)] 
&= \frac{1}{d_A^2}\mathbb E_U\qty[\sum_{k,i} \tr(P_i \otimes P_k U^\dagger M U) \tr(U^\dagger M^\dagger U P_i \times P_k)] \notag \\
&= \tr(M^\dagger M)/d_A^2.
\label{eq:normalization}
\end{align}
This says that under suitable assumptions about the independence of the matrix elements, the matrix $U^\dagger X U d_B / \Vert X \Vert_F$, where $\Vert X \Vert_F = \sqrt{\tr(X^\dagger X)}$, will have an entanglement spectrum distribution given by the MP distribution with $\gamma = \frac{d_B^2}{d_A^2}$. The normalization factor connects the emergence of the universal OES to the independent fluctuation of extensively many Pauli amplitudes. It also implies the extensive growth of the operator entanglement entropy with system size, with the subleading correction given by the entropy of the universal distribution:
\begin{align}
-\sum_{\lambda}\frac{\lambda}{\Vert X \Vert_F^2}\log(\frac{\lambda}{\Vert X \Vert_F^2})
&= -\frac{1}{d_B^2}\sum_{\lambda}\frac{\lambda d_B^2}{\Vert X \Vert_F^2}\log(\frac{\lambda d_B^2}{\Vert X \Vert_F^2})
+ \frac{1}{d_B^2}\sum_{\lambda}\frac{\lambda d_B^2}{\Vert X \Vert_F^2}\log(d_B^2)
\notag\\&\xrightarrow{d\to \infty}S(\text{MP}_\gamma) + \log(d_B^2) \label{eq:oppagecurve}
\end{align}
Since $S(\text{MP}_1) = -1/2$, taking $d_A = d_B$ and $n$ to be the half-system size, the expression for the operator entanglement entropy becomes $n\log 4 - 1/2$, which is the operator Page curve~\cite{page1993average}. As $\gamma \to 0$, we have $S(\text{MP}_\gamma )\to 0$, which recovers the volume-law scaling. 
The emergence of MP statistics seems to be a generic feature of ensembles that are invariant under unitary conjugation, in the sense that if $\{X_d\}_{d\to \infty}$ is a sequence of $d \times d$ matrices, then $\frac{d_B}{\Vert X_d \Vert_F} U_d^\dagger X_dU_d$ concentrates at the MP distribution.
However, we clearly need conditions on $\{X_d\}$ for this to be the case: if $\rho$ is a rank-1 projector, then $\ketbra{\psi}{\psi} = U\rho U^\dagger$ is a \textit{random} rank-1 projector.
The Schmidt values $\{\lambda_i\}$ of $\ket{\psi}$ are MP-distributed, so the OES of $\rho$ is the set of products $\{(\lambda_i\lambda_j)^2\}_{ij}$, whose distribution is written in terms of Bessel functions. 
This suggests that the spectrum of $X$ cannot be too concentrated. 
Another clear counterexample is $X = \mathds 1_{\mathcal H}$, for which $UXU^\dagger$ is always unentangled.
This suggests that the trace of $X$ cannot be too large. 
We conjecture the following conditions:
\begin{conj}
Let $\{X_d\}_{d\to \infty}$ be a sequence of matrices and $\{U_d\}_{d\to \infty}$ a sequence of Haar-random unitary matrices, with $d_B\Vert X_d \Vert_{op}/\Vert X_d \Vert_{F} = \mathcal O(1)$ and $\tr X_d/\Vert X_d \Vert_F \to 0$. Then the OES of $U_d^\dagger X U_d$ concentrates at the MP distribution.
\label{conj:universal_OES_emergence}
\end{conj}
The first condition requires that the eigenvalues of the rescaled matrix do not scale with $d$. 
For the second condition on the trace, since the MP distribution is obtained as a free convolution of the $d_A^2$ projectors onto the columns of $\alpha$ \cite{feier2012methods}, we must only ensure that $\mathcal N\alpha_{00}/d_A = \tr X_d/\Vert X_d \Vert_F \to 0$ almost surely. 
These conditions are supported by numerics, and can be checked directly for the GUE and Haar ensembles, as well as an ensemble of local unitary-evolved operators $X_U = U^{\dagger}(X \otimes \mathds 1_2^{\otimes n-1})U$, where $X$ is a local, traceless observable. 
Establishing or refining this result rigorously is a direction for future research.

\section{OES of classical matrix ensembles}
\subsection{GUE matrices \label{app:GUE}}
In Sec.~\ref{sec:Haar OES} we compare the OES of a chaotically-evolved observable to that of GUE matrices. While the spectrum of a GUE matrix is given by the WD distribution with $\beta = 2$, the OES is demoted to WD with $\beta = 1$. To understand why, we have the following lemma:
\begin{lem}
If $M$ is sampled from the GUE, then $[\alpha_{ij}] \equiv \opbraket{P_i \otimes P_j}{M}$ is a real Wishart matrix.
\end{lem}
\begin{proof}
A GUE matrix $M$ may be sampled by first sampling a matrix $G$ with complex Gaussian entries, then setting $M = \frac{G + G^\dagger}{\sqrt{2}}$. We can first show that the Pauli coefficients of $G$ are i.i.d. complex Gaussians. Let $G_{mn}$ be i.i.d. complex Gaussian variables with variance 1. We write $G = \sum_{j}g_jP_j$, and treat the variables $g_j$ as random variables. Then $g_j = \tr(P_jG) = \sum_{mn}(P_j)_{mn}G_{nm}$. Thus $g_j$ are jointly Gaussian distributed, so we just need to show they are uncorrelated to prove independence:
\begin{equation}
    \langle g_i\overline{g_j}\rangle = \sum_{mnpq}(P_i)_{mn}(P_j)_{qp}\langle G_{nm}\overline{G_{qp}}\rangle = \tr(P_jP_i) = \delta_{ij}
\end{equation}
Since $M$ may be obtained as $M = (G+G^\dagger)/\sqrt{2}$, this shows that the Pauli coefficients of $M$ are independent \textit{real} Gaussian variables, so $[\alpha_{ij}] = \opbraket{P_i \otimes P_j}{M}$ is a real Wishart matrix.
\end{proof}

\subsection{GOE matrices \label{ap:GOE}}
As described in the main text, we need to model the level spacing ratios of random real observables, since $O^tMO$ is real if $O$ is a random superposition-doped automaton and $M$ is a real operator. We model this with the GOE, which is the real part of the GUE. The level spacing ratios clearly do not follow the WD distribution for any $\beta$, but the OES DOS is clearly MP-distributed. This latter fact is explained in the following proposition:
\begin{lem}
If $M$ is sampled from the GOE, then the spectral distribution of $M$ is the same as though $M$ were sampled from the GUE.
\end{lem}
\begin{proof}
Let $M$ be a matrix sampled from the GOE. We have shown above that the GUE is the Ginibre ensemble in Pauli space, and since the GOE is the real part of the GUE, each Pauli coefficient of $M$ corresponding to a complex Pauli operator can be set to zero. The complex structure of the space of operators can be written $\Mat_{d \times d}(\mathbb C) = \mathcal P^{\Re}_d \oplus \mathcal P^{\Im}_d$, where $\mathcal P^{\Re}_d$ are the real Pauli operators, i.e. spanned by Pauli tensor monomials on $\log_2(d)$ qubits having an even number of Pauli-$Y$ factors, and $\mathcal P^{\Im}_d$ similarly denotes operators with an odd number. In terms of this decomposition, $M$ can then be brought into block-diagonal form:
\begin{equation}
M = \mqty(\mathcal P^{\Re}_A \otimes \mathcal P^{\Re}_B & \mathcal P_{A}^{\Re} \otimes \mathcal P_{B}^{\Im} \\ \mathcal P_{A}^{\Im} \otimes \mathcal P_{B}^{\Re} & \mathcal P_{A}^{\Im} \otimes \mathcal P_{B}^{\Im})
\end{equation}
The off-diagonal blocks are complex and the diagonal blocks are real, so the off-diagonal blocks must vanish. The diagonal blocks correspond to two real and independent Gaussian matrices $M_+$ and $M_-$. Since $\mathcal P^{\Re}_d$ spans the space of real symmetric matrices, by comparing dimensions we have $\dim \mathcal P^{\Re}_d = d(d+1)/2$ and $\dim \mathcal P^{\Im}_d = d(d-1)/2$. Let $d_{A/B}^{\pm} \equiv d_{A/B}(d_{A/B}\pm 1)/2$, so that $M_\pm$ is a matrix of dimension $d_A^\pm \times d_B^\pm$. Then the Wishart matrix $M^tM$ factors as $M_+^tM_+ + M_-^tM_-$. In the limit $d\to \infty$, we have $d_{A/B}^\pm \to d_{A/B}^2/2$ and $\Vert M_{\pm}\Vert_F \to \frac{1}{\sqrt{2}}\Vert M \Vert_F$, where $\Vert \cdot \Vert_F$ refers to the Frobenius norm. This shows that a single GOE operator has the same OES statistics as two independent GUE operators of half the size, and the normalization factor expressed in Eq.~\eqref{eq:normalization} is the same for GOE matrices as it is for GUE matrices. 
\end{proof}

With this picture, we can follow the procedure outlined in \cite{atas2013distribution} to derive an approximate expression for the level spacing ratios of a GUE matrix.
\edit{
In a similar spirit to the Wigner surmise, we treat the smallest example that effectively captures the behavior of the larger ensemble;  we compute the joint level spacing distribution of a $2\times 2$ GUE block and a single random Gaussian variable. 
While the true smallest example would be a $4\times 4$ matrix acting on two two-dimensional subsystems, this turns out to be a worse approximation to the true distribution, as we show in Fig.~\ref{fig:GOE_surmise}. One explanation for this is that the real and imaginary block sizes approach each other in the $d \to \infty$ limit.
}
\begin{figure}
\centering
\def\svgwidth{0.45\linewidth}
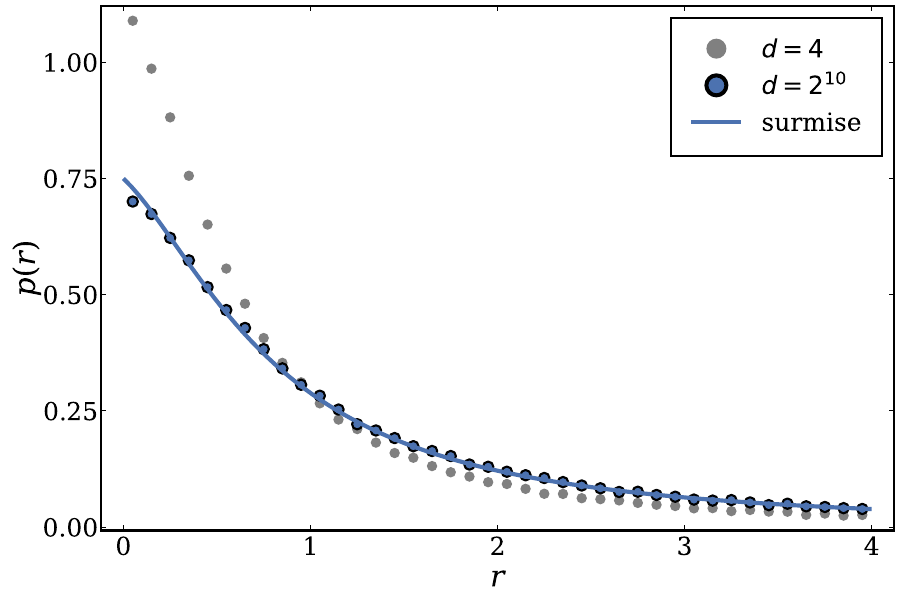
\caption{\edit{Level spacing ratio distribution of $d \times d$ GOE matrices for large $d$ versus $d = 4$ and a comparison to our surmise. Surprisingly, the surmise agrees better with $d = 2^{10}$ than $d = 4$ matrices. The empirical probability densities are obtained by sampling GOE matrices and binning $2^{18}$ entanglement level spacing ratios at both system sizes.}}
\label{fig:GOE_surmise}
\end{figure}

The eigenvalue distribution is given by $p(x,y,z) \propto |x-y|e^{-(x^2+y^2+z^2)/2}$, where $z$ is the eigenvalue originating from the $1\times 1$ Gaussian block. We then compute the level spacing distribution for each possible ordering of the eigenvalues:
\begin{align}
&p_{\text{GOE}}(r) \propto \sum_{x,y,z}\frac{\int dx_1dx_2dx_3 p(x,y,z)\delta(r - \frac{x_1-x_2}{x_2-x_3})}{\int dx_1dx_2dx_3 p(x,y,z)} \ ,
\end{align}
where the sum is taken over possible permutations $x,y,z = \sigma(x_1, x_2, x_3)$, and the bounds of the integral are $x_1 \geq x_2 \geq x_2$. These integrals are simplified by first noting that the GUE eigenvalues are symmetric under exchange. Additionally, following Ref.~\cite{atas2013distribution}, the numerator can easily be determined up to a constant using a clever u-substitution. Putting this together, we find
\begin{equation}
p_{\text{GOE}}(r) = \frac{3}{4}\frac{1+r}{(1+r+r^2)^{3/2}}
\label{eq:pgoe2}
\end{equation}
This distribution emerges in the context of superposition-doped automaton circuits. To illustrate this, we evolve an $R_x(\phi) \equiv e^{-i\frac{\pi}{2}X}$ gate with random Hadamard-doped automaton operators. As apparent in Fig.~\ref{fig:spacingstats}, when $\phi$ is a multiple of $\pi/4$, the spacing statistics track Eq.~\ref{eq:pgoe2}, otherwise the statistics closely follow the WD distribution with $\beta = 1$. The data here are obtained by sampling $128$ matrices on 10 sites
\begin{figure}
    \centering
    \def\svgwidth{0.45\linewidth}
    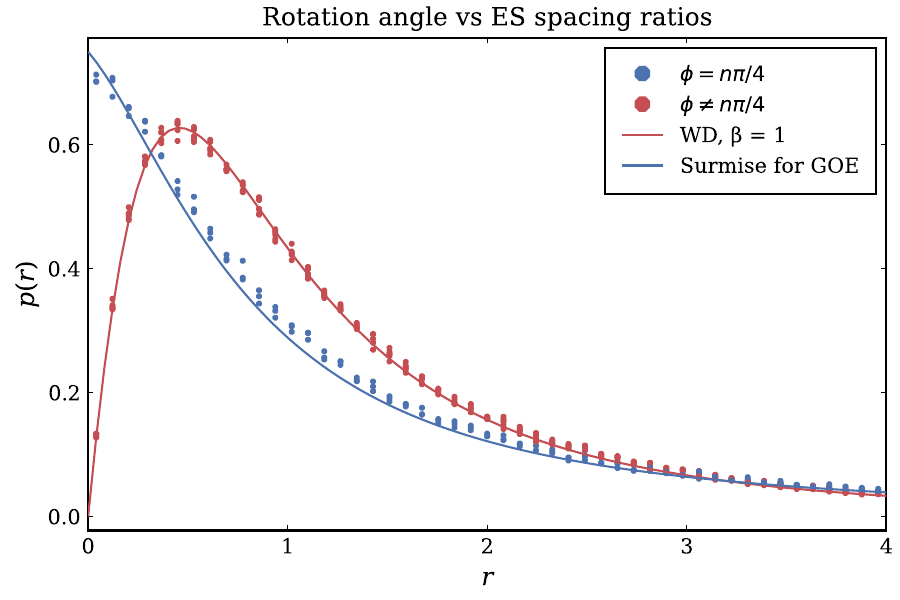
    \caption{The level spacings of $R_x(\phi)$ evolved with a random automaton on ten sites with $n_H = 4$ (chosen randomly) show that although MP is reached for each choice of $\phi$, the level spacings indicate whether or not the evolved operator behaves as a GOE operator with a block-diagonal structure in Pauli-string space, or as a GUE operator, whose level repulsion indicates ergodicity in Pauli-string space.}
    \label{fig:spacingstats}
\end{figure}

\section{Automaton OES}
We have conjectured based on numerical evidence that the automaton OES distribution is the singular-value distribution of Bernoulli matrices. In this section we explore some of the properties of this distribution, including the moments, concentration of the measure, the average Schmidt rank, and the relationship to the topology of random graphs.

\subsection{Bernoulli singular value distribution \label{app:moments}}
In this section, we derive an expression that relates the moments of the average singular value distribution of Bernoulli matrices to a partition-counting problem, which we can then solve computationally to obtain the first few moments of this distribution. In particular, this distribution is uniquely determined by its moments. We show that the moments of the singular value distribution of a single random Bernoulli matrix concentrate at the average moments, providing evidence for our hypothesis that the automaton OES moments also concentrate.

\begin{prop}
Let $P(k)$ denote the set of partitions of $\{1, \dots, k\}$. Let $S_k$ be the set of tuples constructed $S_k \equiv \{(1,1), (2,1), (2,2), (3,2) \dots (k,k), (1, k) \}$. Each pair of partitions $p_i \in  P(k)$ and $p_j \in P(k)$ induces a partition $p_{ij}$ of $S_k$ via the following rule; if $i_1, i_2$ are grouped in $p_i$ and $j_1, j_2$ are grouped in $p_j$, then $(i_1, j_1)$ is grouped with $(i_2, j_2)$ in $p_{ij}$. The $k^{\text{th}}$ moment of the Bernoulli spectral distribution are given by the number of such partitions with $|p_i| + |p_j| - |p_{ij}| = 1$, where $|p|$ denotes the number of groups in the partition $p$.
\end{prop}
\begin{proof}
Let $B$ be a $d\times d$ matrix whose entries are distributed according to $\text{Bern}(\frac{1}{d})$. Then the $k^\text{th}$ moment of the singular value distributions averaged over realizations of $B$ is given by 
\begin{equation}
\mathbb E_B\qty[\frac{1}{d}\tr((BB^\dagger)^k)] = \frac{1}{d}\sum_{\vec i, \vec j \in \{1 ... d\}^{\otimes k}}\mathbb E_B \qty[B_{i_1j_1}B_{i_2j_1}B_{i_2j_2}\dots B_{i_kj_k}B_{i_1j_k}]
\end{equation}
Since $B_{ij}$ takes values $0,1$, we have
\begin{equation}
\mathbb E[B_{i_1j_1}B_{i_2j_1}B_{i_2j_2}\dots B_{i_kj_k}B_{i_1j_k}] = \mathbb E[ B_{m_1n_1}^{ t_1}B_{m_2n_2}^{t_2}\dots B_{m_ln_l}^{t_l}] = \frac{1}{d^l} \ .
\end{equation}
We have grouped the terms in the summand by multiplicity by partitioning the indices $\{(i_t, j_t), (i_{t+1}, j_t)\}_{t=1}^k$ into $l$ parts of size $t_1, \dots, t_l$, where $i_{k+1} \equiv i_1$. Clearly, only partitions with $l$ parts which appear in the sum $d^{l+1}$ times will contribute to the sum as $d \to \infty$. To count the number of times each partition appears, we reindex the sum by a set partition $p_i$ of the indices $\{i_t\}$ and $p_j$ of the indices $\{j_t\}$. The partitions $p_i, p_j$ induce a set partition $p_{ij}$ of $\{(i_1, j_1),(i_{1+1},j_t)\}_{t=1}^{k}$.
Each part of $p_i$ can take on $d-|p_i|-1 \approx d$ values, so the pair $p_i, p_j$ occur in the sum $\approx d^{|p_i|+|p_j|}$ times, where $|p_i|$ refers to the number of parts in $p_i$. The expectation value is $d^{-|p_{ij}|}$, therefore the problem is reduced to counting the set partitions $p_i, p_j$ such that $|p_i|+|p_j|-|p_{ij}| = 1$, where $\vert \cdot \vert$ denotes the number of subsets in the partition. Counting these computationally gives the first few moments: $1, 3, 12, 57, 303, 1747, 10727, 69331$.
\end{proof}

\begin{prop}
The moments of the Bernoulli spectral distribution concentrate in the large-dimension limit, in the sense that
\begin{equation}
\mathbb P_B\qty[\qty|\tr((BB^t)^k) - m_k| > \epsilon] \to 0 \text{ as } d \to \infty
\end{equation}
for any $\epsilon > 0$.
\label{prop:concentration}
\end{prop}
\begin{proof}
We can write the variance explicitly:
\begin{align}
\Var(m_k) &= \frac{1}{d^2}\mathbb E_B\qty[\tr((BB^t)^k)^2] - \frac{1}{d^2}\mathbb E_B\qty[\tr((BB^t)^k)]^2 \notag\\
&= \frac{1}{d^2}\sum_{\vec i, \vec j, \vec k \vec l}\mathbb E_B\qty[\prod_{m}B_{i_mj_m}B_{i_{m+1}j_m}\prod_{q}B_{m_qn_q}B_{m_{q+1}n_q}] \notag\\
&\qquad- \frac{1}{d^2}\sum_{\vec i, \vec j, \vec k \vec l}\mathbb E_B\qty[\prod_{m}B_{i_mj_m}B_{i_{m+1}j_m}]\mathbb E_B\qty[\prod_{q}B_{m_qn_q}B_{m_{q+1}n_q}]
\end{align}
We can then count multiplicities as in the computation of the moments. Let $p_{i\oplus m}$ denote partitions of $2k$ indices $\{i_1, \dots, i_k, m_{1}, \dots, m_k\}$ and $p_{j\oplus n}$ be defined similarly with respect to $\{j_1, \dots, j_k, n_1, \dots, n_k\}$. Let $p_{i\oplus m,j\oplus n}$ be the corresponding partition induced on $\{(i_t, j_t), (i_{t+1}, j_t)\} \cup \{(m_t, n_t), (m_{t+1}, n_t)\}$ according to the rule that $(i_{t_1}, j_{p_1})$ and $(i_{t_2}, j_{p_2})$ are grouped iff $i_{t_1}$ and $i_{t_2}$ are grouped in $p_{i\oplus m}$ and $(i_{t_1}, j_{p_1})$ is grouped with $(m_{t_2}, n_{p_2})$ iff $i_{t_1}$ is grouped with $m_{t_2}$ and $j_{p_1}$ is grouped with $n_{p_2}$ in $p_{j \oplus n}$. By the same argument as the computation of the moments, the first term counts the number of partitions $p_{i\oplus m}$, $p_{j\oplus n}$ such that $|p_{i\oplus m}|+|p_{j \oplus n}| - |p_{i\oplus m, j\oplus n}| = 2$. 

The second term is determined by the same rule as before; we count the tuples of partitions $p_i, p_j, p_m,p_n$ such that both $|p_i| + |p_j| - |p_{ij}| = 1$ and $|p_m| + |p_n| - |p_{mn}| = 1$. Our objective is to show that there is a bijection between these two types of partitions, which will prove the claim.

First, suppose that we have a tuple of partitions $p_i, p_j, p_m,p_n$ such that both $|p_i| + |p_j| - |p_{ij}| =1$ and $|p_m| + |p_n| - |p_{mn}| = 1$. Then clearly we induce a partition $p_{i \oplus m}$ and $p_{j \oplus n}$ by taking the disjoint union of these two partitions, and because of this disjointness, we clearly have $|p_{i\oplus m}| + |p_{j \oplus n}| - |p_{i \oplus m, j\oplus n}| = 2$. Therefore we have an injection from the partitions that contribute to the second sum into the partitions that contribution to the first sum.

On the other hand, we will argue that partitions of this type, i.e. in the image of the mapping above, are the only ones that contribute to the first sum. To see this, notice that any partition $p_{i\oplus m}$ may be obtained by taking the disjoint union of partitions $p_i$ and $p_j$ and grouping together their subsets. We may similarly obtain any $p_{j \oplus n}$ from a $p_j$ and $p_n$. However, we notice that in order to create a group $p_{i\oplus m, j \oplus n}$, we must create groups in $p_{i\oplus m}$ \textit{and} in $p_{j\oplus n}$. The consequence of this is that the grouping operation necessarily strictly decreases $|p_{i \oplus m}| + |p_{j \oplus n}|-|p_{i\oplus m, j \oplus n}|$ from the value obtained via the disjoint union. This shows that the above mapping is surjective, and so we have shown the two are in bijective correspondence. Then the Markov equality shows that the vanishing of the variance implies concentration of the moments.
\end{proof}

Because a Bernoulli matrix can be viewed as the adjacency matrix of a directed graph, the relationship between the OES of random permutations and that of Bernoulli matrices has surprising implications due to known results about the topology of random graphs. The Bernoulli ensemble with $p = 1/d$ can be viewed as the ensemble of random directed graphs on $d$ vertices with a probability $p = 1/d$ of any two vertices being connected by an edge. These graphs, known as Erd\H{o}s-R\'enyi graphs, have been studied previously in mathematics, and the regime $pd = c$ where $c$ is a fixed constant is particularly interesting. Erdos and Renyi observed in \cite{erd6s1960evolution} that there is a striking change in the connectivity of random graphs as the number of edges increases, similar to a percolation transition in classical physics. Let $M = P^{t_{AB}}$ where $P$ is a random permutation. If the graph determined by $M$ is truly random, then the weak connectivity of this graph, defined as the connectivity of the corresponding undirected graph, can be predicted combinatorially. In particular, in the regime $pd = 2$ corresponding to the expected degree of $M+M^t$, there is a single connected component of size $\mathcal O(d)$ \cite{alon2015probabilistic}. The sizes of the smaller components are predicted by a Poisson branching process, $P(|C(v)| = k) = \e^{-2k}2^{k-1}/k!$, where $C(v)$ is the component containing vertex $v$. We examine the graph topology of random permutation matrices numerically, and find excellent agreement with this prediction, as shown in Fig.~\ref{fig:cmp_sizes}. 

\begin{figure}[hbt!]
    \centering
    \def\svgwidth{0.95\linewidth}
    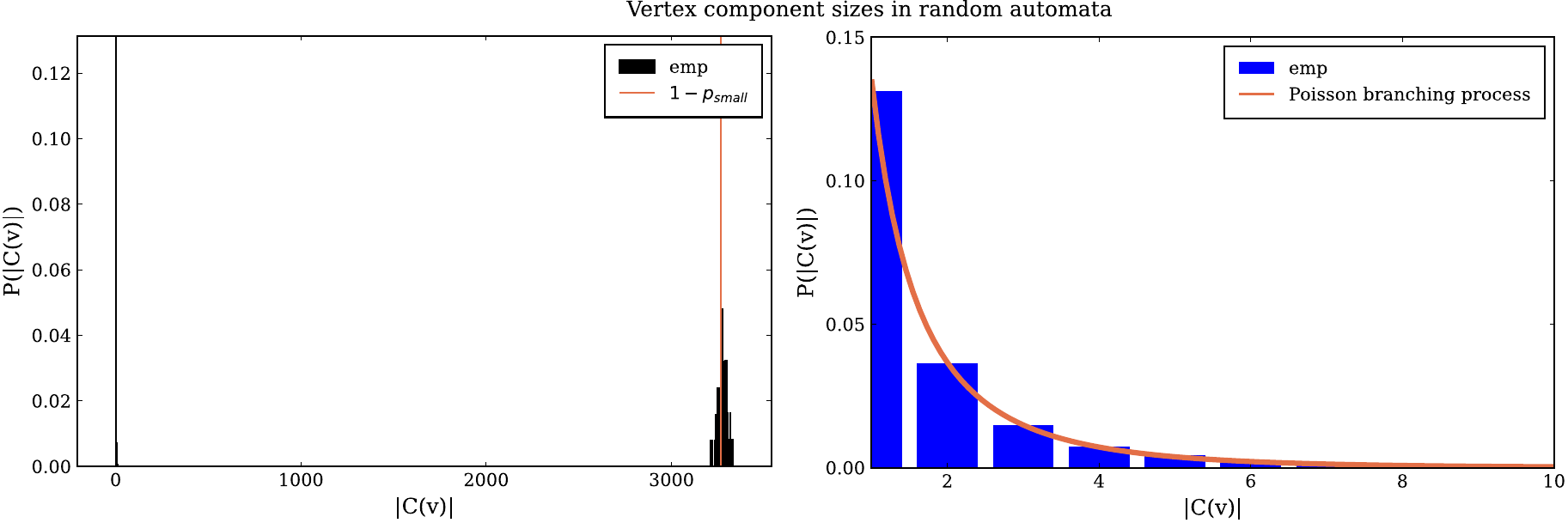
    \caption{Distribution of size of component $C(v)$ containing a vertex $v$. The graphs used are obtained from the adjacency matrix of $P^{t_b}$, where $P \in S_{2^{12}}$ is drawn uniformly at random. The component sizes are averaged over all vertices in each graph, with 100 permutation instances shown. We have $1-p_{small}$ given by the smallest solution to the equation $x = \e^{-2(1-x)}$ and so the largest component has size $\approx 0.80 N$, which is consistent with the random permutations. There is some variation in this number due to the finite size of the permutations. For a random graph, the unique large component would have exactly $(1-p_{small})N$ components. The distribution of small components shown on the right is also consistent with the Poisson branching process.}
    \label{fig:cmp_sizes}
\end{figure}

Because of the non-trivial and size-independent distribution of small components, the entanglement spectrum of a random permutation clearly has a singular component, with atoms at the singular values of all finite components of the graph of all sizes which are low-degree algebraic integers. The small components can be understood as converging to Galton-Watson trees with a Poisson degree distribution of parameter 2 with randomly directed edges \cite{alon2015probabilistic}. However, we also observe that the big component of the graph has a nontrivial singular component as well. We believe this is due to the localization of the eigenvectors of the corresponding Wishart matrix $M M^t$ on subtrees of the graph, which we find numerically for small system sizes. In the study of the spectra of random graphs, the regime of fixed $pd$ is extremely delicate, and obtaining spectral features of these graphs such as the average rank, moments, concentration to a limiting measure, fluctuations, localization of eigenvectors, and continuity properties of the spectrum is an active area of research \cite{guionnet2021bernoulli, bordenave2010rank, bordenave2017mean, salez2015every}. Furthermore, large permutations correspond to random \textit{directed} graphs with average degree $n$, and it is an interesting question whether the singular value distribution of these graphs may be mapped to the spectral distributions of some random undirected graph. Ref.~\cite{salez2015every} has shown that the support of the singular component of the spectral distribution of such a random undirected graph is all of $\mathbb R$, which makes numerical estimation difficult. However, by plotting the weight of eigenvectors corresponding to different eigenspaces, we see that at least for small system sizes, the estimated atoms correspond exclusively to eigenvectors supported on local trees. It seems possible to obtain an exact formula for the moments of the random permutation OES using the generalized Isserlis's formula, counting walks on graphs following \cite{hainzl2024tree}, or by using Weingarten calculus to integrate over the permutation group \cite{collins2022weingarten}.

\subsection{Estimating the Schmidt rank of automata\label{sec:kernelprob}}
The correspondence between random automata and random Bernoulli matrices is a useful computational tool. To derive an asymptotic lower bound on the kernel probability of a random $d \times d$ Bernoulli matrix $B$ with entries $B_{ij} \sim \text{Bern}(1/d)$, we can find the average number of unique nonzero columns in such a matrix. Set $Z_i = \mathds 1\qty[\{\col(i) \neq \col(j), j < i\} \cap \{\col(i) \neq \vec 0\}]$, denoting an indicator function. We find
\begin{align}
    \langle Z_i \rangle &= \sum_{l=1}^{d}\mathbb P\qty(\bigcap_{j < i}\{\col(i) \neq \col(j), j < i\} | \wt(i) = l\})\mathbb P\qty(\{\wt(i) = l\})  \notag\\
    &= \sum_{l=1}^{d}\binom{d}{l}\frac{1}{d^l}\qty(1-\frac{1}{d})^{n-1}\qty(1-\frac{1}{d^l}\qty(1-\frac{1}{d})^{n-l})^{i-1}
\end{align}
Then the average Schmidt rank is given by
\begin{align}
\sum_{i=1}^d \langle Z_i \rangle &= \sum_{l=1}^{d}\binom{d}{l}\frac{1}{d^l}\qty(1-\frac{1}{d})^{d-l}\sum_{i=0}^{d-1}\qty(1-\frac{1}{d^{l}}\qty(1-\frac{1}{d})^{d-l})^{i} \notag \\
&= \sum_{l=1}^{d}\binom{d}{l}\frac{1}{d^l}\qty(1-\frac{1}{d})^{d-l}\frac{1-\qty(1-\frac{1}{d^l}(1-\frac{1}{d})^{d-l})^d}{\frac{1}{d^l}(1-\frac{1}{d})^{d-l}}\notag\\
 &= \sum_{l=1}^{d}\binom{d}{l}\qty(1-\qty(1-\frac{1}{d^l}(1-\frac{1}{d})^{d-l})^d)
\end{align}
We notice that for $l \ll d$,
\begin{align}
\qty(1-(1-\frac{1}{d^l}(1-\frac{1}{d})^{d-l})^{d}) &\to 
\begin{cases}
1-\e^{-\frac{1}{\e}} & l = 1 \\
\frac{d}{\e d^l} & l > 1
\end{cases}
\end{align}
Plugging this in, we find
\begin{align}
 \sum_{l=1}^{d}\binom{d}{l}\qty(1-(1-\frac{1}{d^l}(1-\frac{1}{d})^{d-l})^d) & \approx  d(1-\e^{-\frac{1}{\e}})+\frac{d}{\e}\sum_{l=2}^{d}\binom{d}{l}\frac{1}{d^l} \notag\\
&= d(1-\e^{-\frac{1}{\e}})+\frac{d}{\e}\qty((1+\frac{1}{d})^d - 2) \notag\\
&\to d(2-\frac{2}{\e} - \e^{-\frac{1}{\e}})
\end{align}
And the corresponding lower bound for the average kernel probability 
\begin{align}
p(\{0\}) \approx \frac{2}{\e} + \e^{-\frac{1}{\e}} - 1 \approx 0.427
\end{align}

\subsection{Generic off-diagonal observables\label{app:generic}}
In the main text, we define the related ensemble $\{B_{\theta}\}_{B, \theta}$ by $(B_{\theta})_{mn} =[ B_{mn}\e^{\ii\theta_{mn}}]$ where $\theta_{mn}$ are i.i.d. around the circle and $B$ is a Bernoulli matrix. This is relevant for analyzing the OES of automaton-evolved off-diagonal Hermitian operators. In addition, automaton circuits $P$ are typically defined with phases by $P\ket{n} = \e^{\ii\theta_n}\ket{P(n)}$. We want to show that the universal spectrum found for the Bernoulli ensemble also holds for $\{B_{\theta}\}_{B, \theta}$. Since the method of moments is applicable, we just need to show that the moments of the distribution match.
\begin{prop}
The moments of the average singular value distribution of $\{B_{\theta}\}_{B, \theta}$ are the same as those of the Bernoulli ensemble $\{B\}_B$.
\end{prop}
\begin{proof}
We begin by re-writing the moments as
\begin{equation}
m_k = \mathbb E_B\qty[\frac{1}{d}\Tr\{(B^\dagger B)^k\}] = \mathbb E_B \bra{0} (B^\dagger B)^k \ket{0}
\end{equation}
Since the ensemble $B_{\theta}$ is invariant under permutations when the angles $\theta$ are i.i.d, we can replace the trace with the expectation value of a single state, and we choose $\ket{0}$ arbitrarily. Then we see that for a given realization of $B$, the term $\bra{0} (B^\dagger B)^k \ket{0}$ counts the number of walks on a directed graph of $2k$ steps which begin at $\ket{0}$, then alternate walking forward along the edges of $B$ and backward (because $B^\dagger$ has the edges of $B$ reversed) and ending back at $\ket{0}$. In the walk, each edge traversed is weighted by the corresponding roots of unity $\e^{\ii\theta_n}$ if the step is taken forward and $\e^{-\ii\theta n}$ if taken backward. In the $d \to \infty$ limit, almost all realizations of $B$ are locally treelike at $\ket{0}$, meaning that the subgraph starting at $\ket{0}$ forms a tree. Thus, for every edge which is traversed in the forward direction, picking up a factor of $\e^{\ii\theta_n}$, this edge must also be traversed backward, picking up a factor of $\e^{-\ii\theta_n}$. These factors cancel, and so $m_k$ is the same as for the Bernoulli matrix ensemble.
\end{proof}

\section{Operator dependence \label{app:operator_dependence}}
Like permutations, superposition-doped automaton operators have a preferential basis. However, as the number of Hadamard gates $n_H$ increases, the OES becomes generic regardless of the initial operator. The difference is in a slower convergence when the initial operator is diagonal, as shown in Fig.~\ref{fig:DOS_convergence}.
This could be explained heuristically through the localization in Pauli-string space; the diagonal blocks dominate the KL divergence, and each inserted Hadamard gate creates larger off-diagonal blocks. To formalize this intuition, 
consider $Z(\tau) = \sum_{\vec b}\alpha_{\vec b}Z^{\vec b}$, where the sum runs over binary vectors, $Z^{\vec b} \equiv \bigotimes_{i=1}^{2^n}Z^{b_i}$, and $\alpha_{\vec b}$ are real Gaussian amplitudes. Fix the $n_H$ Hadamard gates on sites labeled by a binary vector $\vec a$. There are $2^{n-n_H}$ choices of $\vec b$ such that $\vec b \cdot \vec a = 0$, so $Z(\tau)$ may be separated into a diagonal component and an off-diagonal component where conjugation by a Hadamard gate turns a Pauli-$Z$ into a Pauli-$X$: 
\begin{equation}
Z(\tau) = \sum_{\vec b \cdot \vec a = 0}^{2^{n-n_H}}\alpha_{\vec b}PZ^{\overline{\vec b \wedge \vec a}}P^t + \sum_{\vec b \cdot \vec a \neq 0}^{2^n(1-2^{-n_H})}\alpha_{\vec b}PZ^{\overline{\vec b \wedge \vec a}}X^{\vec b \wedge \vec a}P^t
\end{equation}
where $\wedge$ denotes the element-wise product (AND) and $\overline{\vec b \wedge \vec a}$ denotes the negation of the components of $\vec b \wedge \vec a$. The first term in the above operator is diagonal, and restricted to the subspace spanned by diagonal Pauli operators. As a crude approximation, we can treat the spectral distribution $p$ as a convex combination $p = ap_X + bp_Z$, where $p_X$ is the spectrum of the first term, $p_Z$ is the spectrum of the second term, $a = 1-2^{-n_H}$, and $b = 2^{-n_H}$. As we observe in Fig.~\ref{fig:DOS_convergence}, $p_X$ quickly approaches the MP distribution, so we approximate $p_X \approx \text{MP}$. Expanding $D_{\text{KL}}(\text{MP} \ \Vert \ \mu)$ to first order in $b$ then gives 
\begin{equation}
D_{\text{KL}}(\text{MP} \ \Vert \ \mu) \approx -\frac{b}{a}-\ln(a) \approx -\frac{b}{a}-b \longrightarrow -2b = 2^{-n_H+1}
\end{equation}
From this argument, we expect $D_{\text{KL}}(\text{MP} \ \Vert \ \mu) \sim 2^{-n_H}$. This is in good agreement with the observed behavior in the lower panel of Fig.~\ref{fig:DOS_convergence}.

\end{document}